\documentclass[a4paper,11pt]{article}

\usepackage[utf8]{inputenc}
\usepackage[T1]{fontenc}
\usepackage[english]{babel}

\usepackage[pdftex]{graphicx}
\usepackage[usenames,dvipsnames]{xcolor}

\usepackage{tikz}
\usetikzlibrary{calc}

\usepackage{amsmath,amssymb,amsfonts,amsthm}
\usepackage{mathabx}
\usepackage[ruled,linesnumbered,vlined]{algorithm2e}

\usepackage{hyperref}
\usepackage{cleveref}

\usepackage{fullpage}
\usepackage{enumitem}

\usepackage{lineno}
\newcommand{\powerset}{\mathcal{P}}
\newcommand{\eps}{\varepsilon}

\newcommand{\ceil}[1]{\left\lceil{#1}\right\rceil}

\newcommand{\abs}[1]{\left | #1 \right |}

\newcommand{\calH}{\mathcal{H}}
\newcommand{\calF}{\mathcal{F}}

\newcommand{\calM}{\mathcal{M}}

\newcommand{\T}{\mathcal{T}}

\newcommand{\E}{\text{E}}
\newcommand{\set}[1]{\left \{ #1 \right \}}
\newcommand{\setbuilder}[2]{\left\{ {#1} \, \middle| \, {#2} \right\}}
\newcommand{\Prp}[1]{\Pr\!\left[{#1} \right]}
\newcommand{\Prpcond}[3][]{\Pr_{{#1}}\left[ {#2} \, \middle| \, {#3} \right]}
\newcommand{\Ep}[1]{\E\!\left[{#1} \right]}
\newcommand{\Epcond}[3][]{\E_{{#1}}\left[ {#2} \, \middle| \, {#3} \right]}

\newcommand{\C}{\mathcal{C}}

\newcommand{\true}[0]{\texttt{true}}
\newcommand{\false}[0]{\texttt{false}}

\renewcommand{\subset}{\subseteq}

\newtheorem{lemma}{Lemma}
\newtheorem{theorem}{Theorem}
\newtheorem{corollary}{Corollary}

\newtheorem{fact}{Fact}
\theoremstyle{definition}
\newtheorem{definition}{Definition}
\newtheorem{remark}{Remark}

\usepackage{todonotes}
\usepackage{xargs}
\newcounter{sideremark}

\newcommand{\newtext}[1]{#1}
\newcommand{\oldtext}[1]{}

\usepackage{authblk}
\title{Fast Similarity Sketching\thanks{A preliminary
    version of this work was presented in the Proceedings of the 58th IEEE Annual Symposium on Foundations of Computer Science, FOCS'17, Berkeley, CA, USA, pages October 15-17, 2017, pages 663--671. This full version has
    stronger results and Jakob Houen as an added author.}}

\author{Søren Dahlgaard}
\author{Jakob Bæk Tejs Houen}
\author{Mathias Bæk Tejs Langhede}
\author{Mikkel Thorup}
\affil{University of Copenhagen\\\texttt{soren.dahlgaard@gmail.com,jakob@tejs.dk,mathias@tejs.dk,}\\\texttt{mikkel2thorup@gmail.com}}
\date{}

\begin{document}

\setcounter{page}{0}
\maketitle
\begin{abstract}
We consider the \emph{Similarity Sketching} problem: Given a universe
$[u] = \{0,\ldots, u-1\}$ we want a random function $S$ mapping
subsets $A\subseteq [u]$ into vectors $S(A)$ of size $t$, such that
the Jaccard similarity $J(A,B) = |A\cap B|/|A\cup B|$ between sets $A$
and $B$ is preserved.  More precisely, define $X_i = [S(A)[i] =
  S(B)[i]]$ and $X = \sum_{i\in [t]} X_i$. We want $\Ep{X_i}=J(A,B)$,
and we want $X$ to be strongly concentrated around $\Ep{X} = t \cdot J(A,B)$
(i.e. Chernoff-style bounds). This is a fundamental problem which has
found numerous applications in data mining, large-scale
classification, computer vision, similarity search, etc. via the
classic MinHash algorithm. The vectors $S(A)$ are also called
\emph{sketches}. Strong concentration is critical, for often we want
to sketch many sets $B_1,\ldots,B_n$ so that we later, for a query set
$A$, can find (one of) the most similar $B_i$. It is then critical
that no $B_i$ looks much more similar to $A$ due to errors in the
sketch.

The seminal \emph{$t\times$MinHash} algorithm uses $t$ random hash functions $h_1,\ldots, h_t$, and stores $\left ( \min_{a\in A} h_1(A),\ldots, \min_{a\in A} h_t(A) \right )$ as the sketch of $A$. The main drawback of MinHash is, however, its $O(t\cdot |A|)$ running time, and finding a sketch with similar properties and faster running time has been the subject of several papers. Addressing this, Li et al.~[NIPS'12] introduced \emph{one permutation hashing (OPH)}, which creates a sketch of size $t$ in $O(t + |A|)$ time, but with the drawback that possibly some of the $t$ entries are ``empty'' when $|A| = O(t)$. One could argue that sketching is not necessary in this case, however the desire in most applications is to have \emph{one} sketching procedure that works for sets of all sizes. Therefore, filling out these empty entries is the subject of several follow-up papers initiated by Shrivastava and Li~[ICML'14]. However, these ``densification'' schemes fail to provide good concentration bounds exactly in the case $|A| = O(t)$, where they are needed. 

In this paper we present a new sketch which obtains essentially the best of both worlds. That is, a fast $O(t\log t + |A|)$ expected running time while getting the same strong concentration bounds as $t\times$MinHash. Our new sketch can be seen as a mix between sampling with replacement and sampling without replacement. We demonstrate the power of our new sketch by considering popular applications in large-scale classification with linear Support Vector Machines (SVM) as introduced by Li et al.~[NIPS'11] as well as approximate similarity search using the Locality Sensitive Hashing (LSH) framework of Indyk and Motwani [STOC'98].

\end{abstract}

\thispagestyle{empty}

\newpage
\setcounter{page}{1}
\section{Introduction}
In this paper we consider the following problem which we call the
\emph{similarity sketching} problem. Given a large key universe
$[u] = \{0,\ldots,u-1\}$ and positive integer $t$ we want a random function $S$
mapping subsets $A\subseteq [u]$ into vectors (which we will call sketches)
$S(A)$ of size $t$, such that similarity is preserved. More precisely,
we consider the \emph{Jaccard
similarity} $J(A,B) = |A\cap B|/|A\cup B|$ between sets $A$ and $B$. Define
$X_i = [S(A)[i] = S(B)[i]]$ for each $i\in [t]$, where $S(A)[i]$ denotes the $i$th entry of the vector $S(A)$ and $[x]$ is
the Iverson bracket notation with $[x]=1$ when $x$ is true and $0$ otherwise.
We want $\Ep{X_i}=J(A,B)$ for each $i\in t$. Moreover,
we want $X = \sum_{i\in [t]} X_i$ to be strongly concentrated
around $\Ep{X} = t \cdot J(A,B)$. That is, the sketches can be used to estimate
$J(A,B)$ by doing a pair-wise comparison of corresponding
entries. We will call this
the \emph{alignment property} of the similarity sketch.
Strong concentration, with error probabilities dropping exponentially in $t$ like with Chernoff-bounds, is critical. Often we want
to sketch many sets $B_1,\ldots,B_n$ so that we later, for a query set
$A$, can find (one of) the most similar $B_i$. It is then critical
that no $B_i$ looks much more similar to $A$ due to errors in the
sketch.

The standard solution to the similary sketching problem is $t\times$MinHash
algorithm\footnote{\url{https://en.wikipedia.org/wiki/MinHash}}. The algorithm
works as follows: Let $h_0,\ldots, h_{t-1} :
[u]\to [0,1]$ be random hash functions and define $S(A) = (\min_{a\in
A} h_0(a), \ldots, \min_{a\in A} h_{t-1}(a))$. This corresponds to sampling
$t$ elements from $A$ with replacement and thus has all the above desired
properties.

MinHash was originally introduced by Broder et
al.~\cite{Broder97,BroderGMZ97,BCFM00} for the AltaVista search engine and
has since been used as a standard tool in many applications including duplicate
detection~\cite{BroderGMZ97,Henzinger06}, all-pairs
similarity~\cite{BayardoMS07}, large-scale learning~\cite{LiSMK11},
computer vision~\cite{ShakhnarovichDI08}, and similarity
search~\cite{IndykM98}.
The application \cite{LiSMK11} to Support Vector Machines (SVMs)
in Machine Learning is an instructive example of the use of
the alignment property. The basic idea is that we want the similarity between sets to be grow with a dot-product between associated
vectors. To get a bit vector, \cite{LiSMK11} suggests
that for each coordinate, we hash to get a single bit and based
on this bit, replace the coordinate with two bits: ${\tt 01}$ or ${\tt 10}$. Now,
for the dot-product, if we had a match in a coordinate, we get
a match in both bits, adding two to the dot-product. If we did not
have a match, then with probability 1/2,
either both or no bits will match. Dissimilarity is thus essentially halved
in the reduction. The more important thing is that more similar
sets are expected to get higher dot-products, and this is the main point for the SVM applications. Mathematically, a cleaner alternative is to use
the hash bit to replace the the coordinate with $-1/\sqrt t$ or
$1/{\sqrt t}$. Now, in expectation, the dot-product is exactly the Jaccard similarity.

The main drawback of $t\times$MinHash is the $O(t\cdot |A|)$ running
time that we pay because we have to find the minimum in $A$ with $t$
independent hash functions. To appreciate the scale of the different
parameters, let us just consider the classic application from
\cite{Broder97,BroderGMZ97,MRS08} for text similarity. For each text,
they look at the set of $w$-shingles\footnote{$w$-shingles are also referred to as
$w$-grams, e.g., when used for finding similarity between DNA sequences (see
\texttt{wikipedia.org/wiki/N-gram}).} where a $w$-shingle is a tuple
of $w$ consecutive words, e.g., \cite{Broder97,BroderGMZ97} use $w=4$.
The similarity between texts are the similarity between their sets of
$w$-shingles. With roughly $10^5$ common english words, the number of
possible 4-shingles is $u\approx 10^{20}$.  Also, note that the number
of distinct $4$-shingles in a large text can be close to the text size
even though the number of distinct words is much smaller, so the sets
we consider can be quite large. 

While the sketch size $t$ is typically
meant to be much smaller than the set size, it can still be sizeable,
e.g., \cite{LiSMK11} suggests using $t = 500$ and \cite{Li15} suggests
using $t=4000$. From a more theoretical perspective, if
we use Locallity Sensitive Hashing (LSH) for set similarity
among $n$ sets, then $t=\Omega(n^\rho)$ where $\rho$ is a constant
parameter. We shall return to this application later.

If we do not care about alignment, then an alternative to
$t\times$MinHash sketching is the Bottom-$t$ sketch described
in~\cite{Broder97,CK07,thorup13bottomk}. The idea is to use a single
hash function and just store the $t$
smallest hash values from $A$ (so $1\times$MinHash=Bottom-$1$). We can find the bottom-$t$ sketch of $A$ in $O(|A|)$ time. However, with $t>1$,
there is no alignment
between the $t$ sketch values from different sets, and the alignment is
needed for applications in LSH as well as for Support Vector Machines (SVM)
that we will also discuss later.

Bachrach and Porat~\cite{BachrachP13} suggested a more
efficient way of computing $t\times$MinHash values with $t$ different hash
functions. They use $t$ different polynomial hash functions that are related,
yet pairwise independent, so that they can systematically maintain the MinHash
for all $t$ polynomials in $O(\log t)$ time per element of $A$. There are two
issues with this approach: It is specialized to work with polynomials and
MinHash is known to have constant bias unless the polynomials considered have
super-constant degree~\cite{patrascu10kwise-lb}, and this bias does not decay
with independent repetitions. Also, because the experiments are only pairwise
independent, the concentration is only limited by Chebyshev’s inequality and
thus nowhere near the Chernoff bounds we want for many applications.

Another direction introduced by Li et al.~\cite{LiOZ12} is \emph{one
  permutation hashing} (OPH) which works by hashing the elements of
$A$ into $t$ buckets and performing a MinHash in each bucket using the
same hash function. While this procedure gives $O(t + |A|)$ sketch
creation time it also may create empty buckets and thus only obtains a
sketch with $t'\le t$ entries when $|A| = o(t\log t)$. One may argue
that sketching is not even needed in this case. However, a common goal
in applications of similarity sketching is to have \emph{one}
sketching procedure which works for all set size -- one data structure
that works for an entire collection of data sets of different sizes in
the case of approximate similarity search. It is thus very desirable
that the similarity sketch works well independently of the size of the
input set.

Motivated by this, several follow-up
papers~\cite{ShrivastavaLiUAI14,ShrivastavaLiICML14,Shrivastava17}
have tried to give different schemes for filling out the empty entries
of the OPH sketch (``densifying'' the sketch). These papers all
consider different ways of copying from the full entries of the sketch
into the empty ones. Due to this approach, however, these
densification schemes all fail to give good concentration guarantees
when $|A|$ is small, which is exactly the cases in which OPH gives
many empty bins and densification is needed. This is because of the
fundamental problem that unfortunate collisions in the first round
cannot be resolved in the second round when copying from the full
bins. To understand this consider the following extreme example: Let
$A$ be a set with two elements. Then with probability $1/t$ these two
elements end in the same bin where only one survives (the one with the
smallest hash value). Densification will then just copy the surviver
to all $t$ entries. Such issues may lead to very poor similarity
estimation and this is illustrated with experiments in
\Cref{fig:histogram}. A further issue is that the state-of-the-art
densification scheme of Shrivastava~\cite{Shrivastava17} has a running
time of $O(t^2)$ when $|A|= O(t)$.

\begin{figure}[htbp]
    \centering
    \includegraphics[width=.6\textwidth]{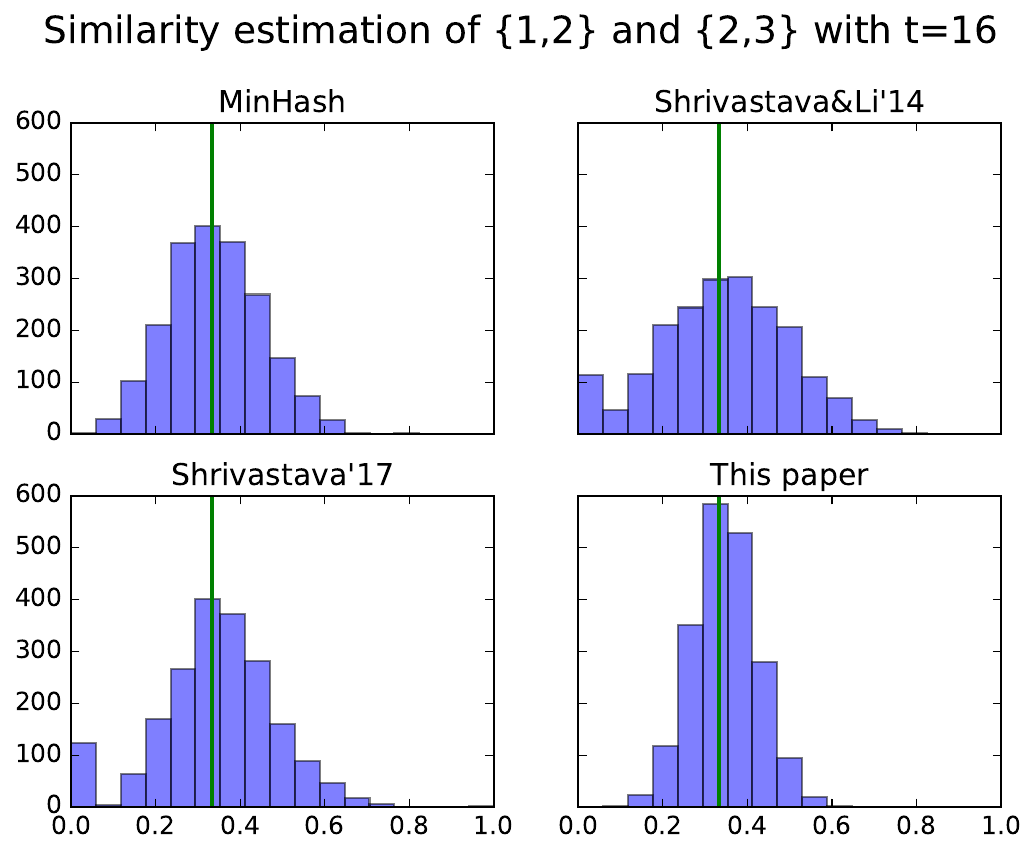}
    \caption{Experimental evaluation of similarity estimation of the sets $A =
    \{1,2\}$ and $B = \{2,3\}$ with different similarity sketches and $t=16$.
    Each experiment is repeated 2000 times and the $y$-axis reports the
    frequency of each estimate. The green line indicates the actual similarity.
    The two methods based on OPH perform poorly as
    each set has a probability of $1/t$ to be a single-element sketch. Our new
    method outperforms $t\times$MinHash as it has an element of ``without replacement''.}
    \label{fig:histogram}
\end{figure}

\subsection{Our contribution}
In this paper we obtain a sketch which essentially obtains the best of
both worlds. That is, strong concentration guarantees for similarity
estimation as well as a fast expected sketch creation time of $O(t\log
t + |A|)$. On top of that, our new sketching algorithm is simple and easy
to implement. 

Our new sketch can be seen as a mixture between sampling
with and without replacement and in many cases outperforms
$t\times$MinHash. An example of this can be seen in the toy example of
\Cref{fig:histogram}, where the ``without replacement''-part of our
sketch gives better concentration compared to MinHash.  Our sketch can
be employed in places where $t\times$MinHash is currently
employed to improve the running time from the ``quadratic'' $O(t\cdot
|A|)$ to the ``near-linear'' $O(t\log t + |A|)$. In this paper we
focus on two popular applications, which are large-scale learning with
linear SVM and approximate similarity search with LSH. Before
discussing these applications, we describe
our result  more formally.

\begin{theorem}\label{thm:prob}\footnote{This powerful theorem
    was not in our original conference version \cite{DahlgaardKT17}.}
        Let $[u] = \set{0,1,2,\ldots,u-1}$ be a set of keys and let $t$ be a
        positive integer. There exists an algorithm that given a set $A \subset
        [u]$ in expected time $O\left(\abs{A}+t \log t\right)$ creates a size-$t$
        vector $v(A)$ of non-negative real numbers with the following properties.

        For two sets $A,B\subseteq [u]$ with Jaccard similarity $J(A, B) = J$ it holds that $v(A \cup B)_i =
        \min\set{v(A)_i,v(B)_i}$ for each index $i \in [t]$. For $i \in [t]$
        let $X_i = 1$ if $v(A)_i = V(B)_i$ and $0$ otherwise. Let $I \subset [t]$
        be a subset of $k$ indices. Let $j \in [t] \setminus I$, then,
        \begin{align*}
            \min\left(J, \frac{tJ - \sum_{i \in I} X_i}{t - k}\right)
                \le \Prpcond{X_j = 1}{\sigma\left( (X_i)_{i \in I} \right)}
                \le \max\left(J, \frac{tJ - \sum_{i \in I} X_i}{t - k}\right)
            \, .
        \end{align*}
    \end{theorem}

    If we sampled with replacement then the probability that $X_j = 1$ when conditioning
    on $(X_i)_{i \in I}$ is exactly $J$, and if we sampled without replacement then
    the probability that $X_j = 1$ when conditioning on $(X_i)_{i \in I}$ is exactly
    $\frac{tJ - \sum_{i \in I} X_i}{t - k}$. Thus the theorem shows qualitatively that
    our new sketch falls between sampling with replacement and sampling without replacement.

    Now a nice implication of the theorem is that we get classic Chernoff
    concentration bounds.
    \begin{corollary}
        \label{cor:mainThmIntro}
        Use the same setting as \Cref{thm:prob} and let $X = \frac{1}{t}
        \sum_{i \in [t]} X_i$. Then $\Ep{X} = J$ and for every $\delta
        > 0$ it holds that:
        \begin{align*}
            \Prp{X \ge J(1+\delta)}
            & \le 
            \left ( \frac{e^
            \delta}{(1+\delta)^{1+\delta}} \right )^{t J}
            \, ,
            \\
            \Prp{X \le J(1-\delta)}
            & \le 
            \left ( \frac{e^{-\delta}}{(1-\delta)^{1-\delta}} \right )^{t J}
            \, .
        \end{align*}
    \end{corollary}
    The above theorems assume access to fully random hashing (this
    is also assumed for the standard MinHash sketch).
However, we will also show how to implement our sketch using 
mixed tabulation hashing (introduced by Dahlgaard et
al.~\cite{DahlgaardKRT15}) which is practical and can be evaluated in
$O(1)$ time. The concentration bounds get slightly weaker and more
complicated, as they do with any realistic hashing scheme.

\subsection{Speeding up Locality Sensitive Hashing on Sets}\label{sec:lsh_intro}
One of the most powerful applications of the $t\times$MinHash algorithm
is for the \emph{approximate set similarity search problem} using
\emph{Locality Sensitive Hashing (LSH)}. For this application we will
crucially need both the alignment and strong probability
bounds from Theorem \ref{thm:prob}. Our fundamental contribution will be to improve the time
bound, reducing the effect of the set size from multiplicative to
additive. This is very important when dealing with larger sets, e.g.,
the above mentioned examples where the sets are the 4-shingles of
texts \cite{Broder97,BroderGMZ97,MRS08}.

While our new similarity sketch forms the base of our improvement, we
need several other ideas to address other challenges, particularly
when it comes to high probability results. To describe our
contribution, we first need to revisit the Locality Sensitive Hashing
framework introduced by Indyk and Motwani~\cite{IndykM98}. It is a
general framework based on the following type of data-independent hash
families:
\begin{definition}[Locality sensitive hashing~\cite{IndykM98}]\label{def:lsh}
    Let $(X, S)$ be a similarity space and let $\calH$ be a family of hash
    functions $h : X \to R$. We say that $\calH$ is
    $(s_1, s_2, p_1, p_2)$-sensitive if for any $x, y \in X$ and $h \in \calH$
    chosen uniformly random we have that
    \begin{itemize}
        \item If $S(x, y) \ge s_1$ then $\Prp{h(x) = h(y)} \ge p_1$.
        \item If $S(x, y) \le s_2$ then $\Prp{h(x) = h(y)} \le p_2$.
    \end{itemize}
\end{definition}
The Locality Sensitive Hashing framework takes a $(s_1, s_2, p_1,
p_2)$-sensitive family $\calH$ for a similarity space $(X, S)$, and
uses it to solve the Approximate Similarity Search problem for a stored
set $Y\subseteq X$ with parameters $0 < s_2 < s_1 < 1$. That is, given
a query $q \in X$ it returns an element $a \in Y$ with $S(a, q) \ge
s_2$, if there exists an element $b \in Y$ with $S(b, q) \ge s_1$ with
constant probability. 
Note that the definition of the locality
sensitive hash function $\calH$ is oblivious to the concrete stored
set $Y$.

In this paper we only focus on the 
Jaccard set similarity with stored
set of sets $Y=\calF$.  
The classical way of using the LSH framework with respect to the
Jaccard similarity is using the seminal MinHash algorithm, which was
originally introduced by Broder et
al.~\cite{Broder97,BroderGMZ97}. The MinHash algorithm is defined as
follows: Given a family $\calH$ of hash functions $h : U \to R$ we
define a new family $\calH^{\min}$ of hash functions $h^{\min} :
\powerset(U) \to U$ where $h^{\min}(A) = \arg\min_{x \in A} h(x)$ for
any $A \subseteq U$. If the domain $R$ is large enough, we can assume
that there are no collisions, hence that $h^{\min}$ is well-defined. If $h
\in \calH$ is chosen uniformly at random we get that $\Prp{h^{\min}(A) =
  h^{\min}(B)} = J(A, B)$ for any $A, B \subseteq U$. This shows that
$\calH^{\min}$ is a $(j_1, j_2, j_1, j_2)$-sensitive family for any $0
< j_2 < j_1 < 1$.

The next idea from~\cite{IndykM98} is to create sketches using $K$ functions from $\calH^{\min}$.
For a set $A$, we get the sketch $S(A) = (h_{0}^{\min}(A), \ldots, h_{K - 1}^{\min}(A))$ where $h_{0}, \ldots, h_{K - 1} \in \calH$. If the hash functions
are independent, then $\Prp{S(A) = S(B)} = J(A, B)^K$, so $S$
is $(j_1, j_2, j_1^K, j_2^K)$-sensitive. Setting 
$K = \ceil{\frac{\log(n)}{\log(1/j_2)}}$, we
get that $S$ is  $(j_1, j_2, O(1/n^\rho),1/n)$-sensitive where
$\rho= \frac{\log(1/j_1)}{\log(1/j_2)}$. While the above construction
may seem a bit ad-hoc, O'Donnell et al.~\cite{ODonnellWZ14} 
have shown it to be optimal in the sense that 
we cannot in general construct a $(j_1, j_2, o(1/n^\rho),1/n)$-sensitive 
family of hash functions.
With $n=\abs{\calF}$ stored sets, we only
expect a constant number of false matches $S(A)=S(Q)$ where $A\in\calF$ 
and $J(A,Q)< j_2$.

To get a positive match with constant probability, we use
$L=\ceil{n^\rho}$ sketch functions $S_0,\ldots,S_{L-1}$. If we have a
set $B$ with $J(Q,B)\geq j_2$, then with constant probability, there
is an $i\in[L]$ such that $S_i(B)=S_i(Q)$. To create a data structure,
for each $i$, we have a hash table that with each sketch value $s$ stores
pointers to all sets $A\in \calF$ with $S_i(A)=s$. When we query a set
$Q$ we compute the sketch $S_i(Q)$ and
lookup all the matches $A\in \calF$ with
$S_i(A)=S_i(Q)$ using the hash tables for each $i\in [L]$.
We check the matches one by one to see if
$J(A,Q)\leq j_2$, stopping if a good one is
found. This data structure has constant 1-sided error probability.
It uses $O(n \cdot L + \sum_{A \in \calF} \abs{A})
= O(n^{1 + \rho} + \sum_{A \in \calF} \abs{A})$ space and
$O(L \cdot K \cdot \abs{Q}) = O(n^\rho \log n \cdot \abs{Q})$ query time.
More precisely, the query time is dominated by
two parts: 
\begin{itemize}
\item[(1)]
We have to compute $O(L \cdot K)$ hash values for similarity sketches which
takes $O(L \cdot K \cdot \abs{Q})$ time using $O(L \cdot K)\times$MinHash.
\item[(2)] In expectation the data structure returns $O(L)$ false positives
which have to be filtered out. This takes $O(L \cdot \abs{Q})$ time.
\end{itemize}
Different techniques has been used to speed up the query time and mostly
the focus has been on improving the dominant part (1). Andoni and Indyk~\cite{AndoniI06} looked at
the general LSH framework and limited the number of evaluations of locality
sensitive hash functions. The idea is to create the sketches by combining smaller
sketches together. More precisely, a much smaller collection of sketches of size
$m = o(L)$ is created and then every $\binom{m}{t}$ combination of sketches is
formed. This technique is also known as tensoring. In the context of Jaccard
similarity it improved part (1) of the query time from 
$O(L\cdot K\cdot\abs{Q})$ to $O(L\cdot\abs{Q})$, matching the
bound for part (2). Thus they got an overall query time of 
$O(L\cdot\abs{Q})=O(n^\rho \cdot \abs{Q})$.

\paragraph{Contribution}
Our original conference version \cite{DahlgaardKT17} had overlooked
\cite{AndoniI06} and focussed directly on improving the original bounds
in (1) and (2).  Concerning (1), we note that the analysis assumes that
all the values in the $O(L \cdot K)\times$MinHash are independent.
This is not the case for our new similarity sketch from Theorem
\ref{thm:prob} with $t=O(L \cdot K)$, but it turns out that the
probability bounds from Theorem \ref{thm:prob} do suffice. This
implements (1) in
$O(t\log t+\abs{Q})=O(L\cdot K\cdot \log n + \abs{Q})$.
For a better implementation, we can first create an
intermediate sketch of size $O(\log^2(n))$ and then sample the $L$
sketches of size $K$ from this intermediate sketch. This improves part
(1) of the query time to $O(K\cdot L + \abs{Q})$.

Improving (2) requires some different ideas, but already in
\cite{DahlgaardKT17}, we improved part (2)
of the query time to $O(L + \abs{Q})$.

Christiani~\cite{Christiani17} noticed that our construction
also composes nicely with tensoring technique~\cite{AndoniI06}
so that we only need $L$ instead of $L\cdot K$ sketch values.
As a result, part (1) of the query time is improved
from $O(K\cdot L + \abs{Q})$
to $O(L + \abs{Q})$, matching the time for part (2). Hence
the total query time becomes $O(L + \abs{Q})=O(L + \abs{Q})=
O(n^\rho + \abs{Q})$, which is the natural target for constant
error probability. This is the combined solution presented
in the current full paper.

Often we
want a smaller error probability $\eps > 0$, e.g., $\eps=1/n$.
The generic standard approach is to use $O(\log(1/\eps))$ independent data
structures, each failing with constant probability, and then return
the best solution found, if any. Indeed this is suggested by Motwani and
Indyk in \cite[p. 605]{IndykM98} and \cite[p. 327]{HIM12}. 
Since the data structures can use a common representation
of the sets, we would get a space usage of
$O(n^{1 + \rho} \log(1/\eps) + \sum_{A \in \calF} \abs{A})$
and a query time of $O((n^\rho + \abs{Q}) \log(1/\eps))$.

Here we further improve the query time to
$O\left( n^\rho \log(1/\eps) + \abs{Q}\right)$
while having the same space usage. Thus we preserve our optimal linear
dependence on $\abs{Q}$ even for high probability results. Adding
it all up, we prove the following result for the approximate
set similarity search problem using LSH:
\begin{theorem}\label{thm:LSH}
We are given a family $\calF$ of up to $n$ sets from a large universe $U$. Moreover,
we are given two constant parameters $j_1$ and $j_2$ with $0 < j_2 < j_1 < 1$, as
well as an error parameter $\eps > 0$ which may be subconstant.
% We assume
% $\log(1/\eps)\leq\min\{w,e^{n^{o(1)}}\}$ where $w$ is the word-length.

We present an Approximate Similarity Search data structure
with 1-sided error probability $\eps$: given a query set $Q\subset U$,
if there is a set $B \in \calF$ with $J(B, Q) \geq j_1$, the data structure
returns a set $A \in \calF$ with $J(A, Q) \geq j_2$ with probability $1-\eps$.
Moreover, if $A$ is returned, it always has $J(A, Q) \geq j_2$.
With $\rho = \frac{\log(1/j_1)}{\log(1/j_2)}$ our data structure uses
    $O\left(n^{1+\rho} \log(1/\eps)  + \sum_{A \in \calF}\abs{A}\right)$ space and it has
    query time $O\left( n^\rho \log(1/\eps) + \abs{Q}\right)$.
\end{theorem}

\subsection{Notation}
For a real number $x$ and an integer $k$ we define $x^{\underline{k}} = x(x-1)(x-2) \ldots (x-k+1)$.
For an expression $P$ we let $[P]$ denote the variable that is $1$ if $P$ is true
and $0$ otherwise. For a non-negative integer $n$ we let $[n]$ denote the set
$[n] = \set{0,1,2,\ldots,n-1}$.

\section{Fast Similarity Sketching}\label{sec:sketch}
In this section we present our new sketching algorithm, which takes a set
$A\subseteq [u]$ as input and produces a sketch $S(A,t)$ of size $t$. When $t$
is clear from the context we may write just $S(A)$.

Our new similarity sketch is simple to describe: Let $h_0,\ldots, h_{2t-1}$ be
random hash functions such that for $i\in [t]$ we have $h_i : [u]\to [t]\times
[i,i+1)$ and for $i\in \{t,\ldots, 2t-1\}$ we have $h_i : [u]\to \{i-t\}\times
[i,i+1)$. For each hash function $h_i$ we say that the output is split into a
bin, $b_i$, and a value, $v_i$. That is, for $i\in [2t]$ and $a\in [u]$ we have
$h_i(a) = (b_i(a), v_i(a))$, where $b_i(a)$ and $v_i(a)$ are restricted as
described above. We may then define the $j$th entry of the sketch $S(A)$ as
follows:
\begin{equation}\label{eq:sketchdef}
    S(A)[j] = \min\!\left\{v_i(a)\mid a\in A, i\in [2t], b_i(a) = j\right\}\ .
\end{equation}
In particular, the hash functions $h_t,\ldots, h_{2t-1}$ ensure that each entry
of $S(A)$ is well-defined. Furthermore, since we have $v_i(a) < v_j(b)$ for any
$a,b\in [u]$ and $0\le i < j < 2t$ we can efficiently implement the sketch
defined in \eqref{eq:sketchdef} using the Similarity-Sketch procedure in \Cref{alg:sketch}. A bin $S[b]$ gets ``filled'' the first time it gets assigned a
value $<\infty$ in line \ref{line:hash}. The algorithm terminates when all bins are filled, and if this happens in the first round with $i=0$, then
the sketch created is identical to that of \emph{one permutation
  hashing}~\cite{LiOZ12}.

\begin{algorithm}
    \caption{\FuncSty{Similarity-Sketch}}
    \label{alg:sketch}
    \DontPrintSemicolon
    \SetKwInOut{Input}{input}\SetKwInOut{Output}{output}
    \Input{$A$, $t$, $h_0,\ldots, h_{2t-1}$}
    \Output{The sketch $S(A,t)$}
    \BlankLine
    $S\gets \infty^t$\;
    $c\gets 0$\;
    \For{$i=0,\cdots,2t-1$}{
        \For{$a\in A$}{
            $b,v\gets h_i(a)$\label{line:hash}\;
            \If{$S[b] = \infty$}{
                $c\gets c+1$\;
            }
            $S[b]\gets \min(S[b], v)$\;
        }
        \If{$c = t$}{
            \Return $S$\;
        }
    }
\end{algorithm}

We will start our analysis of $S(A)$ by bounding the running time of
\Cref{alg:sketch}.
\begin{lemma}\label{lem:sketch_time}
    Let $A\subseteq [u]$ be some set and let $t$ be a positive integer. Then
    the expected running time of \Cref{alg:sketch} is $O(t\log t + |A|)$.
\end{lemma}
\begin{proof}
  We always have a trivial worst-case upper bound of $O(t\cdot |A|)$ and
  this is $O(t\log t)$ if $|A|=O(\log t)$. Otherwise, we may assume
  $|A| \geq 2\log t$. Fix $i$ to be the smallest
            value such that $|A|\cdot i \geq 2\cdot t\log t$. Then $i\leq t$, and then the
            probability that a given bin is empty after evaluating
            $h_0,\ldots, h_{i-1}$ is at most
            \[
                (1 -1/t)^{|A|\cdot i} \le (1-1/t)^{2\cdot t\log t} \le
                1/t^2\ .
            \]
            It follows that the probability of any bin being empty is at most
            $1/t$ and thus the expected running time is $O(|A|\cdot i +
            \frac{|A|\cdot t}{t}) = O(t\log t + |A|)$.
\end{proof}
Next, we will prove several properties of the sketch. The first is an
observation that the sketch of the union of two sets can be computed solely
from the sketches of the two sets.

\begin{fact}
    Let $A,B$ be two sets and let $t$ be a positive integer. Then
    \[
        S(A\cup B,t)[i] = \min\!\left(S(A,t)[i], S(B,t)[i]\right)\,.
    \]
\end{fact}

The main technical lemma regarding the sketch is \Cref{lem:prob} below.
The lemma show that our sketch can qualitatively be seen as a mixture
between sampling with replacement and sampling without replacement.
We will use this lemma to show that we get an
unbiased estimator as well as Chernoff-style concentration bounds.
\begin{lemma}\label{lem:prob}
	Let $A, B$ be sets with Jaccard similarity $J(A, B) = J$, let $t$ be a
    positive integer For each $i \in [t]$ let
	$X_i = \left[ S(A, t)[i] = S(B, t)[i] \right]$. Let $I \subset [t]$ be a
	subset of $k$ indices. Let $j \in [t] \setminus I$ then
	\begin{align*}
		\min\left(J, \frac{tJ - \sum_{i \in I} X_i}{t - k}\right)
			\le \Prpcond{X_j = 1}{\sigma\left( (X_i)_{i \in I} \right)}
			\le \max\left(J, \frac{tJ - \sum_{i \in I} X_i}{t - k}\right)
		\, .
	\end{align*}
\end{lemma}
\begin{proof}
	Define $\T = (T_0, T_1, \ldots, T_{2t - 1})$ in the following way. Let
	$T_0 = b_0 \left( A \cup B \right)$ and for $i \ge 1$ let
	$T_i = b_i \left( A \cup B \right) \setminus \left( T_0 \cup \ldots \cup T_{i - 1} \right)$.
	Assume in the following that $\T$ is fixed. It clearly suffices to prove
	this theorem for all possible choices of $\T$. Let $n = \abs{A \cup B}$,
	then $nJ = \abs{A \cap B}$.

	We will prove the claim when the set $I$ is chosen uniformly random
	among the subsets of $[t]$ of size $k$, and where $j$ is chosen uniformly
	random from $[t] \setminus I$. Because of symmetry this will suffice.

	The probability that $j \in T_{h}$ is
	$\frac{\abs{T_h} - \abs{T_h \cap I}}{t - k}$. Conditioned on $j \in T_{h}$
	the probability that $X_j = 1$ is exactly
	$\frac{nJ - \sum_{i \in T_h \cap I} X_i}{n - \abs{T_h \cap I}}$.
	So the probability that $X_j = 1$ is
	\[
		p = \Ep{\sum_{h \in [2t]}
			\frac{\abs{T_h} - \abs{T_h \cap I}}{t - k}
		\cdot
			\frac{nJ - \sum_{i \in T_h \cap I} X_i}{n - \abs{T_h \cap I}}
		}
	\]
	If we fix $\abs{T_h \cap I}$ then $T_h \cap I$ is a uniformly random subset
	of $I$ of size $\abs{T_h \cap I}$. This implies that
	\[
		\Epcond{\sum_{i \in T_h \cap I} X_i}{\sigma\left( (X_i)_{i \in I}, \abs{T_h \cap I} \right)}
			= \abs{T_h \cap I}\frac{\sum_{i \in I} X_i}{k}
	\]
	and that
	\begin{align*}
		&\sum_{h \in [2t]}
			\Epcond{
				\frac{\abs{T_h} - \abs{T_h \cap I}}{t - k}
			\cdot
				\frac{nJ - \sum_{i \in T_h \cap I} X_i}{n - \abs{T_h \cap I}}
			}{\sigma\left( (X_i)_{i \in I}, \abs{T_h \cap I} \right)}
	\\
		&\qquad= \sum_{h \in [2t]}
				\frac{\abs{T_h} - \abs{T_h \cap I}}{t - k}
			\cdot
				\frac{nJ - \abs{T_h \cap I}\frac{\sum_{i \in I} X_i}{k}}{n - \abs{T_h \cap I}}
	\end{align*}

	If $J \le \frac{tJ - \sum_{i \in I} X_i}{t - k}$ then
	$J \ge \frac{\sum_{i \in I} X_i}{k}$ which implies that
	\[
		J 
			\le \frac{nJ - \abs{T_h \cap I}\frac{\sum_{i \in I} X_i}{k}}{n - \abs{T_h \cap I}}
			\le \frac{\abs{T_h}J - \abs{T_h \cap I}\frac{\sum_{i \in I} X_i}{k}}
					 {\abs{T_h} - \abs{T_h \cap I}}
	\]
	and inserting these estimates give us
	\begin{align*}
		J 
			&= \sum_{h \in [2t]} \frac{\abs{T_h} - \abs{T_h \cap I}}{t - k} \cdot J
		\\
			&\le \sum_{h \in [2t]} \frac{\abs{T_h} - \abs{T_h \cap I}}{t - k} 
			\cdot \frac{nJ - \abs{T_h \cap I}\frac{\sum_{i \in I} X_i}{k}}{n - \abs{T_h \cap I}}
		\\
			&\le \sum_{h \in [2t]} \frac{\abs{T_h} - \abs{T_h \cap I}}{t - k} 
			\cdot \frac{\abs{T_h}J - \abs{T_h \cap I}\frac{\sum_{i \in I} X_i}{k}}
					   {\abs{T_h} - \abs{T_h \cap I}}
		\\
			&= \frac{tJ - \sum_{i \in I} X_i}{t - k}
	\end{align*}
	Similar calculations shows that if $J \ge \frac{tJ - \sum_{i \in I} X_i}{t - k}$
	then
	\begin{align*}
		J 
			\ge \sum_{h \in [2t]} \frac{\abs{T_h} - \abs{T_h \cap I}}{t - k} 
			\cdot \frac{nJ - \abs{T_h \cap I}\frac{\sum_{i \in I} X_i}{k}}{n - \abs{T_h \cap I}}
			\ge \frac{tJ - \sum_{i \in I} X_i}{t - k}
	\end{align*}
	which finishes the proof.
\end{proof}

As a corollary we immediately get that the estimator is unbiased.
\begin{lemma}
    Let $A,B$ be sets with Jaccard similarity $J(A,B) = J$ and let $t$ be a
    positive integer. Let $X_i = [S(A,t)[i] = S(B,t)[i]]$ and let $X =
    \sum_{i\in [t]} X_i$. Then $\Ep{X} = tJ$.
\end{lemma}
\begin{proof}
    This follows directly by applying \Cref{lem:prob} with $k=0$.
\end{proof}
We also get nice bounds on the moments even when conditioning on a subset of the indices.
This lemma will be important in \Cref{sec:lsh} where we will use it to prove our result
of improving LSH.
\begin{lemma}\label{lem:moment}
	Let $A, B$ be sets with Jaccard similarity $J(A, B) = J$, let $t$ be a
    positive integer. Let $I \subseteq [t]$ be a
	set of $k$ indices and $K \subseteq [t] \setminus I$ another disjoint set of
	$m$ indices. Then
	\begin{align*}
		\left( \frac{tJ - k}{t - k} \right)^m \le \Epcond{\prod_{j \in K} X_j}{\sigma\left( (X_i)_{i \in I} \right)} 
			&\le \left( \frac{tJ}{t - k} \right)^m
		\, .% \\
		% \Epcond{\prod_{j \in K} \left( 1 - X_j \right)}{\sigma\left( (X_i)_{i \in I} \right)} 
		% 	&\le \left( \frac{t\left( 1 - J \right)}{t - k} \right)^m
		% \, .
	\end{align*}
\end{lemma}
\begin{proof}
	The proof for the lower bound is completely analogous to proof for the
    upper bound so we only show the arguments for the upper bound.

    First we note that if $k \ge t(1 - J)$ then $\frac{tJ}{t - k} \ge 1$ and
    the result is trivially true. So in the rest of the proof we assume that
    $k \le t(1 - J)$

    We will prove that for any subset $K' \subset K$ and $h \in K \setminus K'$
    then the following is true
    \[
        \Epcond{X_h}{\sigma\left( (X_i)_{i \in I} \right) \wedge \bigwedge_{j \in K'} \left( X_j = 1 \right)}
            \le \frac{tJ}{t - k}
    \]
    This is easily seen by using \Cref{lem:prob}
    \[
        \Epcond{X_h}{\sigma\left( (X_i)_{i \in I} \right) \wedge \bigwedge_{j \in K'} \left( X_j = 1 \right)}
            \le \max\left(
                J,
                \frac{tJ - |K'| - \sum_{i \in I} X_i}{t - |K'| - k}
            \right)
            \le \frac{tJ}{t - k}
    \]

    Now we enumerate the elements of $K = \{ v_0, \ldots, v_{m - 1} \}$ and
    see that
    \begin{align*}
        \Epcond{\prod_{j \in [m]} X_{v_j}}{\sigma\left( (X_i)_{i \in I} \right)}
            = \prod_{j \in [m]}
                \Epcond{ X_{v_j} }
                        {\sigma\left( (X_i)_{i \in I} \right) \wedge \bigwedge_{h \in [j]} \left( X_{v_h} = 1 \right)}
            \le \left( \frac{tJ}{t - k} \right)^m
    \end{align*}
\end{proof}
Finally, we also get Chernoff-style concentration bounds as follows.
\begin{lemma}
	\label{lem:chernoffLemma}
	Let $A,B$ be sets with Jaccard similarity $J(A,B) = J$ and let $t$ be a
    positive integer. Let $X_i = [S(A,t)[i] = S(B,t)[i]]$ and let $X =
    \sum_{i\in [t]} X_i$. Then for $\delta > 0$
	\begin{align*}
		\Prp{X \ge J(1+\delta)}
		& \le 
		\left ( \frac{e^\delta}{(1+\delta)^{1+\delta}} \right )^t
		\, ,
		\\
		\Prp{X \le J(1-\delta)}
		& \le 
		\left ( \frac{e^{-\delta}}{(1-\delta)^{1-\delta}} \right )^t
		\, .
	\end{align*}
\end{lemma}
\begin{proof}
	The upper bound follows from \Cref{lem:moment} with $k = 0$ and \cite[Corollary 1]{SchmidtSS95} since
	Chernoff bounds are derived by bounding $\Ep{e^{\lambda X}}$ for some $\lambda > 0$.
	
	The lower bounds follows from considering $Y = \sum_{i \in [t]} Y_i$ where $Y_i = 1 - X_i$
	and $Y = t-X$. Since $Y_i = [S(A \cup B,t)[i] = S((A \cup B) \setminus (A \cap B),t)[i]]$
	we can use the same argument as for the upper bound, see \cite[Page 4]{SchmidtSS95}.
\end{proof}

\paragraph{Practical implementation}
In \Cref{alg:sketch} we used $2t$ fully random hash functions to
implement our new similarity sketch. We now briefly describe how to
avoid this requirement by instead using just \emph{one} Mixed
Tabulation hash function as introduced by Dahlgaard et
al.~\cite{DahlgaardKRT15}. We do not present the entire details, but
refer instead to the theorems of \cite{DahlgaardKRT15} which can be
used directly in a black-box fashion. We note that Dahlgaard et
al.~\cite{DahlgaardKRT15} did address \emph{one permutation
  hashing}~\cite{LiOZ12} which is identical to our similarity
sketch if all bins are filled in the first round with $i=0$.

In tabulation-based hashing we view each key, $x\in [u]$, as a vector
$(x_0,\ldots, x_{c-1})$ of $c$ characters, where each $x_i\in
[u^{1/c}]=\Sigma$, and $\Sigma$ is called the \emph{alphabet
  size}. The space will be $O(c|\Sigma|)$ and hash values are computed
in in $O(c)$ time. As in Dahlgaard et al.~\cite{DahlgaardKRT15}, we
need $|\Sigma|\geq \delta\cdot t\log t$ for some sufficiently large
constant $\delta$. The output of tabulation hashing is a bit-string
that we can easily split into two parts, one for the bin and one for
the value.

In our similarity sketch from \Cref{alg:sketch}, we use $2t$
independent hash fuctions $h_i$. Here we view the index $i$ as an
extra most significant character from $[2t]\subset \Sigma$. We then
use a single mixed tabulation hash function $H$ taking indexed keys
from $\Sigma^{c'}$ for $c'=c+1$. With with original key
$a=(a_0,\ldots,a_{c-1})$, we compute $h_i(a)$ in \Cref{line:hash} of
\Cref{alg:sketch} as $H(i,a_0,\ldots,a_{c-1})$, but fixing the
bin to $i-t$ if $i\geq t$,

We now consider two cases:
\begin{itemize}
\item Suppose $|A| \le |\Sigma|/2$ and let $i\leq 2t$ be an integer
  such $i|A|\leq |\Sigma|/2$. Then we have at most $|\Sigma|/2$ indexed
  keys in $[i]\times A$. It now follows from \cite[Theorem
    1]{DahlgaardKRT15}, that w.h.p., the indexed keys from $[i]\times A$ hash
  fully randomly, just like in our original analysis. If $2t|A|\leq|\Sigma|/2$,
  we pick $i=2t$, and then this implies full randomness over all indexed keys.
  Otherwise, we pick $i=\min\{t,|\Sigma|/(2|A|)\rfloor\}$.
  Then $i|A|\geq |\Sigma|/(4|A|)$. Now,
    as in the proof of \Cref{lem:sketch_time}, the probability that a
  given bin is empty after $i$ rounds in \Cref{alg:sketch} is at most
  \[(1-1/t)^{|A|\cdot i}\leq (1-1/t)^{(\delta t\log t)/4}<1/t^{\delta/4}.\]
  For a large enough $\delta$, we conclude that all bins are filled,
  w.h.p., hence that we have full randomness over all indexed keys
  considered before termination.
\item Suppose now that $|A| >|\Sigma|/2$. In particular this implies
  that we have a subset $A'$ of $|\Sigma|/2$ keys. As in the previous
  case, $[1]\times A'$ get hashed fully randomly, filling all bins,
  w.h.p., but this must then also be the case for $[1]\times A$. Thus,
  w.h.p., \Cref{alg:sketch} terminates at the end of the first round,
  meaning that it behaves like \emph{one permutation
    hashing}~\cite{LiOZ12}. In this case both correctness and running
  time follows immediately from \cite[Theorem 2]{DahlgaardKRT15}.
\end{itemize}
We note that concentration bounds from \cite{DahlgaardKRT15} for mixed
tabulation have been later been improved by Houen and Thorup
in \cite[\S 1.6]{HT22:chaos}. This
does not, however, change the above description of how we would apply mixed
tabulation in our similarity sketch.

\oldtext{\subsection{Separation}

It can be useful to check if the Jaccard similarity of two sets are above a
certain threshold or not, without having to actually calculate the Jaccard
similarity. Specifically, we assume that we are given two sets $A$ and $B$
and want to determine if $J(A,B) \ge \gamma$. Intuitively, this should be
easy if $J(A,B)$ is either much larger or much smaller than $\gamma$ and
difficult when $J(A,B) \approx \gamma$. Inspired by Henzinger and Thorup
\cite{HenzingerT97providebound} we consider the following algorithm
for doing so: We let $t \ge r$ be positive integers and let
$X_i = [S(A,t)[i] = S(B,t)[i]]$ for $i \in [t]$. We now run a for loop
with an index $i$ going from $r$ to $t$. At each step we check if
$\sum_{j < i} X_j \le i \cdot \gamma + \sqrt[3]{i^2}$. If so the algorithm terminates
and returns \false{}. If no such $i$ is found the algorithm returns
\true{}. See \Cref{alg:separate} for pseudo-code.

\begin{algorithm}
    \caption{\FuncSty{Separate}}
    \label{alg:separate}
    \DontPrintSemicolon
    \SetKwInOut{Input}{input}\SetKwInOut{Output}{output}
    \Input{$t,(X_0,X_1,\ldots,X_{t-1}),r,\gamma$}
    \Output{\true{} or \false{}}
    \BlankLine
%    $S = X_0 + X_1 + \ldots + X_{r-2}$\;
    $S = 0$\;
    \For{$i=1,2,\ldots,t$}{
    	$S = S + X_{i-1}$\;
    	\If{$i \ge r$ $\And$ $S \le i \cdot \gamma + \sqrt[3]{i^2}$}{
    		\Return \false\;
    	}
    }
    \Return \true\;
\end{algorithm}

Assume that we use \Cref{alg:separate} with $X_i = [S(A,t)[i] = S(B,t)[i]]$.
In \Cref{lem:separate} we show how the algorithm behaves when
$J(A,B) \ge \gamma + \delta$ and $J(A,B) \le \gamma - \delta$ respectively.
Furthermore, if we only count the running time of the algorithm in case
the algorithm returns \false{} the expected used time is $O(r)$.
If $J(A,B) \ge \gamma + \delta$, $\delta$ is a constant and $r$ is a
sufficiently large constant (depending on $\delta$) then the algorithm
returns \true{} with constant probability. If $J(A,B) \le \gamma - \delta$
and $\delta$ is a constant then the algorithm returns \true{} with probability
exponentially small in $t$.

\begin{lemma}
	\label{lem:separate}
	%Let $A,B$ be sets, $t \ge r$ be positive integers and $\gamma, \delta \in [0,1]$.
	%Let $X_i = [S(A,t)[i] = S(B,t)[i]]$ for $i \in [t]$. Assume that we run
	%\Cref{alg:separate} with parameters $(t,(X_0,\ldots,X_{t-1}),\gamma,r)$.

	Let $t \ge r$ be integers and $\gamma, \delta, p \in [0,1]$. Let $X_0,X_1,\ldots,X_{t-1}$
	be independent variables with values in $[0,1]$ such that $\Ep{X_i} = p$ for every
	$i \in [t]$. Assume that we run \Cref{alg:separate} with parameters
	$(t,(X_0,\ldots,X_{t-1}),\gamma,r)$.
	
%	Let $\tau$ be the running time of the algorithm, and let $\tau_F = \tau$
%	if the algorithm returns \false{} and let $\tau_F = 0$ otherwise. Then
%	$\Ep{\tau_F} = O(r)$.

	Let $\tau$ be the number of iterations of the for loop during the algorithm, and let
	$\tau_F = \tau$ if the algorithm returns \false{} and let $\tau_F = 0$ otherwise.
	Then $\Ep{\tau_F} = O(r)$.
	
	%If $J \ge \gamma + \delta$ and $r \ge \frac{8}{\delta^3}$ the algorithm returns
	If $p \ge \gamma + \delta$ and $r \ge \frac{8}{\delta^3}$ the algorithm returns
	\true{} with probability at least
	\begin{align}
		\label{eq:probLargeJ}
		1 - \frac{e^{-\delta^2 r/2}}{1-e^{-\delta^2/2}}
		\, .
	\end{align}
	
	%If $J \le \gamma - \delta$ the algorithm returns \true{} with probability at most
	If $p \le \gamma - \delta$ the algorithm returns \true{} with probability at most
	\begin{align}
		\label{eq:probSmallJ}
		e^{-2\delta^2 t}
		\, .
	\end{align}
\end{lemma}
\begin{proof}
See \Cref{sec:separateproof}.
\end{proof}}

\oldtext{\section{Speeding up LSH}
We consider the approximate similarity search problem with parameters $0 < j_2
< j_1 < 1$ on a collection, $\mathcal{C}$, of $n$ sets from a large universe $[u]$.
We will create a data-structure similar to the LSH structure as described in
\Cref{sec:lsh_intro} with parameters $L$ and $K$. That is, for each set
$A\in\mathcal{C}$ (and query $Q$) we will create $L$ sketches $S_0(A),\ldots,
S_{L-1}(A)$ of size $K$ such that for any two sets $A,B\subseteq [u]$ and $i\in
[L]$ we have the following property:
\begin{itemize}
    \item $\Prp{S_i(A) = S_i(B)} \le J(A,B)^K$.
    \item If $J(A,B) \ge j_1$ then $\Prp{S_i(A) = S_i(B)} = \Theta(J(A,B)^K)$.
\end{itemize}
By setting $K = \ceil{\frac{\log n}{\log(1/j_2)}}$ and $L = \ceil{(1/j_1)^K}$ and using
the analysis of \cite{IndykM98} this immediately gives us $O(n^{1+\rho} +
\sum_{A\in\mathcal{C}}|A|)$ space usage and $O(n^\rho\log n + T(n^\rho,\log
n,|Q|))$ expected query time, where $T(L,K,z)$ is the time it takes to create $L$
sketches of size $K$ for a set of size $z$. By providing a more efficient
way to compute the sketches $S_i(A)$ we thus obtain a faster query
time.

In order to create the sketches $S_0(A),\ldots, S_{L-1}(A)$ described above we
first create a $L\times K$ table $T$ such that for each $i\in [L]$
and $j\in [K]$ we have $T[i,j]$ is a uniformly random integer chosen
from $\{j\cdot t/K,\ldots, (j+1)\cdot t/K - 1\}$, where $t$ is a parameter
divisible by $K$ to be chosen later. The rows of the matrix are chosen
independently. Each row is filled using a $2$-independent source of randomness.
Now for a given $A\in [u]$ (or $Q$) we do as follows:
\begin{enumerate}
    \item Let $S(A)$ be a size $t$ similarity sketch of \Cref{sec:sketch}.
    \item For each $i \in [L]$ and $j\in [K]$ let $S_i(A)[j]$ be
        $S(A)[T[i,j]]$.
\end{enumerate}
It follows that the time needed to create $S_0(A),\ldots, S_{L-1}(A)$ for any
$A\in [u]$ is $O(LK + t\log t + |A|)$.
We let $t = K \cdot \ceil{1 + K \cdot \left ( \frac{1}{j_1} - 1 \right )}$.

We start by bounding the number of ``false positives''.
\begin{lemma}
	\label{lem:falsePositiveProb}
	Let $A \in \C$ be such that $J(A,Q) \le j_2$. Then for any $i \in [L]$ the
	probability that $S_i(A) = S_i(Q)$ is at most $\frac{1}{n}$.
\end{lemma}
\begin{proof}
Fix $T[i,j]$ and let $v_j = T[i,j]$ for all $j \in [K]$. Now define $(X_j)_{j \in [t]}$
as in \Cref{lem:moment}. Then $S_i(A) = S_i(Q)$ if and only if
$X_{v_j} = 1$ for all $j \in [K]$, i.e. if $\prod_{j \in [K]} X_{v_j} = 1$. By
    \Cref{lem:moment}
this happens with probability at most $(J(A,Q))^K \le j_2^K \le 1/n$.
\end{proof}
\Cref{lem:falsePositiveProb} shows that for each $i \in [L]$ the expected
number of sets $A \in \C$ with $J(A,Q) \le j_2$ and $S_i(A) = S_i(Q)$ is at
most $\abs{\C} \cdot \frac{1}{n} = 1$. Thus, the expected number of pairs
$(i,A) \in [L] \times \C$ with $J(A,Q) \le j_2$ and $S_i(A) = S_i(Q)$ is at
most $L$.

Let $A_0 \in \C$ be a set such that $J = J(A_0,Q) \ge j_1$. We will give a
lower bound on the probability that there exists an index $i \in [L]$ such that
$S_i(A_0) = S_i(Q)$. For $i \in [L]$ let $Y_i = [S_i(A_0) = S_i(Q)]$ and let $Y
= \sum_{i \in [L]} Y_i$. Using \Cref{lem:moment} and
the same reasoning as in \Cref{lem:falsePositiveProb} we see that
$\Ep{Y_i} \ge \frac{(Jt)^{\underline{K}}}{t^{\underline{K}}}$. Using this we
get:
\begin{align*}
	\Ep{Y_i} \ge 
	\frac{(Jt)^{\underline{K}}}{t^{\underline{K}}}
	>
	\left ( \frac{tJ-K}{t-K} \right )^K
	=
	J^K \cdot \left (
		1 - \frac{K(1-J)}{J(t-K)}
	\right )^K
	\, .
\end{align*}
By the definition of $t$ we have that $\frac{K(1-J)}{J(t-K)} \le \frac{1}{K}$.
Hence $\Ep{Y_i} \ge J^K \cdot \left ( 1 - \frac{1}{K} \right )^K \ge J^K/4$.
As a consequence we get that $\Ep{Y} \ge L \cdot J^K/4 \ge L \cdot j_1^K/4 \ge
\frac{1}{4}$, i.e. that the expected number of indices $i \in [L]$ such that
$S_i(A_0) = S_i(Q)$ is $\Omega(1)$. However, this does not suffice that such an
index exists with constant probability. In order to prove this we will bound
$\Ep{Y^2}$ and use the inequality $\Prp{Y > 0} \ge
\frac{(\Ep{Y})^2}{\Ep{Y^2}}$, which follows from Cauchy-Schwarz's Inequality.

\begin{lemma}
	\label{lem:productBound}
	Let $i_0,i_1 \in [L]$ be different indices. Then
	\begin{align*}
		\Ep{Y_{i_0} Y_{i_1}} 
		\le 
		J^{2K}
		\cdot \left (
			1 + \frac{K(1-J)}{Jt}
		\right )^K
		\, .
	\end{align*}
\end{lemma}
\begin{proof}
The values $(T[i,j])_{(i,j) \in \set{i_0,i_1} \times [L]}$ are all independent by definition.
Let $R$ be the set containing these value, i.e.
\begin{align*}
	R = \set{T[i,j] \mid (i,j) \in \set{i_0,i_1} \times [L]}
	\, .
\end{align*}
Define $X_j = [S(A_0)[j] = S(Q)[j]]$ as in \Cref{lem:moment} and
fix the value of $R$. Then by \Cref{lem:moment} $\Ep{Y_{i_0}Y_{i_1} \mid R} \le J^{\abs{R}}$.
It remains to understand $\abs{R}$. For $j \in [K]$ let $Z_j = [T[i_0,j] \neq T[i_1,j]]$. Then
it is easy to see that $\abs{R} = K + \sum_{j \in [K]} Z_j$, that $(Z_j)_{j \in [K]}$ are independent
and that $\Prp{Z_j = 1} = 1 - \frac{K}{t}$. Hence we can upper bound $\Ep{Y_{i_0} Y_{i_1}}$ by
\begin{align*}
	\Ep{Y_{i_0} Y_{i_1}}
	\le 
	\Ep{
	J^K
	\prod_{j \in [K]} J^{Z_j}
	}
	=
	J^K
	\prod_{j \in [K]} \Ep{J^{Z_j}}
	\, .
\end{align*}
Now 
\begin{align*}
	\Ep{J^{Z_j}} =
	\left ( 1 - \frac{K}{t} \right ) \cdot J + \frac{K}{t}
	=
	J \cdot \left (
		1 + \frac{K(1-J)}{Jt}
	\right )
	\, ,
\end{align*}
and therefore $\Ep{Y_{i_0} Y_{i_1}} \le J^{2K} \cdot \left ( 1 + \frac{K(1-J)}{Jt} \right )^K$.
\end{proof}
By the definition of $t$ we have $1 + \frac{K(1-J)}{Jt} \le 1+\frac{1}{K} \le e^{1/K}$, and therefore
\Cref{lem:productBound} gives that $\Ep{Y_{i_0} Y_{i_1}} \le e J^{2K}$. Hence:
\begin{align*}
	\Ep{Y(Y-1)} = 
	\sum_{i_0, i_1 \in L, i_0 \neq i_1} \Ep{Y_{i_0} Y_{i_1}}
	\le 
	L(L-1) \cdot e \cdot J^{2K}
	<
	L^2 \cdot e \cdot J^{2K}
	\, .
\end{align*}
Since $\Ep{Y} \le L \cdot J^K$ we get that $\Ep{Y^2} \le L \cdot J^K + L^2 \cdot e \cdot J^{2K}$.
So the probability that $Y > 0$ can be bounded below as follows:
\begin{align*}
	\Prp{Y > 0}
	\ge 
	\frac{(\Ep{Y})^2}{\Ep{Y^2}}
	\ge 
	\frac{(L \cdot J^K)^2/16}{e (L \cdot J^K)^2 + (L \cdot J^K)}
	=
	\frac{1}{16(e + (L \cdot J^K)^{-1})}
	\ge 
	\frac{1}{16(e + 1)}
	=
	\Omega(1)
	\, .
\end{align*}

\paragraph{Avoiding false positives}
We let $M = \set{(i,A) \in [L] \times \C \mid S_i(A) = S_i(Q)}$ be the set of matches. We
have proved that for each $A_0 \in \C$ with $J(A_0,Q) \ge j_1$ with probability $\Omega(1)$
there exists $i \in [L]$ such that $(i,A_0) \in M$. Furthermore, we have proved that the
expected number of pairs $(i,A) \in M$ with $J(A,Q) \le j_2$ is at most $L$.
Naively, we could go through all the elements in $M$ until we find $(i,A) \in M$ such that
$J(A,Q) > j_2$ in $O(\abs{Q})$ time per pair. The expected time would be $O \left ( L \cdot \abs{Q} \right )$,
since in expectation we would check $\le L$ pairs $(i,A)$ with $J(Q,A) \le j_2$.

In order to obtain a expected running time of $O \left ( L \cdot \abs{Q} \right )$ we do
something different. We split it into two cases depending on whether $\abs{M} \ge C L$ 
or $\abs{M} \le C L$ for some sufficiently large constant $C$ depending on $j_1,j_2$.
We can in $O(L)$ time check if $\abs{M} \ge CL$. First assume that $\abs{M} \ge CL$. Then 
we find a subset $M' \subset M$ of size $\abs{M'} = \ceil{CL}$, which we can clearly do
in $O(L)$ time. Then we sample a uniformly random pair $(i,A) \in M'$ and check if $J(A,Q) > j_2$.
By Markov's inequality the number of pairs $(i,A) \in M$ with $J(A,Q) \le j_2$ is at most
$\frac{CL}{2}$ with probability $\ge 1 - \frac{2}{C}$, and in this case we find a set $A$
with $J(A,Q) \ge j_2$ with probability at least $\frac{1}{2}$. The time used in this case
is clearly $O(L + \abs{Q})$.

Now assume that $\abs{M} \le CL$. We assume that we have made a similarity sketch of size
$\Theta\left(\log n \right )$ for each set $A \in C$ and $Q$ - the running time and space
usage for this is clearly dominated by what is used for the sketch of size $t$.
For each $(i,A) \in M$ we now use \Cref{alg:separate} with $\gamma = \frac{j_1+j_2}{2}$
on this sketch to separate $J(A,Q)$. We choose $r$ to be a sufficiently large constant.
If the algorithm returns \true{} we calculate $J(A,Q)$ and if it returns \false{} we
discard $A$. We note that for any set $A$ with $J(A,Q) \le j_2$ the probability that
we the algorithm returns \true{} is at most $\frac{1}{n}$. Hence the expected number of sets $A \in \C$
with $J(A,Q) \le j_2$ for which we calculate $J(A,Q)$ explicitly is at most $O(1)$.
We conclude that the running time is $O(L + \abs{Q})$, since if we calculate $J(A,Q)$ for
a set $A$ with $J(A,Q) > j_2$ we can terminate the algorithm and return $A$.
Furthermore, if there exists a set $A_0 \in \C$ with $J(A_0,Q) \ge j_1$ the probability
that the algorithm returns \true{} is $\Omega(1)$ since $r$ is sufficiently large and 
so there is probability $\Omega(1)$ of finding a set with Jaccard similarity $> j_2$ in this case.

If there exists a set $A_0$ with Jaccard similarity $J(A_0,Q) \ge j_1$ we conclude
that the probability of finding a set $A$ with $J(A,Q) > j_2$ is therefore at least
$\Omega(1) - \frac{2}{C}$. By choosing $C$ sufficiently large we ensure that 
$\Omega(1) - \frac{2}{C} = \Omega(1)$.

Summarizing we get the following theorem.
\begin{theorem}\label{thm:lsh}
	Let $0 < j_2 < j_1 < 1$ be constants, and let $\rho = \frac{\log(1/j_1)}{\log(1/j_2)}$.
	Let $U$ be a set of elements and let $\C$ be collection of $n$ sets from $U$.
	Then there exists a data structure using space
	$O \left ( n^{1+\rho} + \sum_{A \in \C} \abs{A} \right )$ and has query time
	$O \left (n^{\rho}\log n + \abs{Q} \right )$ such that: Given a set $Q$ if
	there exists a set $A_0 \in \C$ with $J(A_0,Q) \ge j_1$, then with constant
	probability the data structure returns a set $A \in \C$ with $J(A_0,Q) > j_2$.
\end{theorem}
}

\newtext{\section{Speeding up LSH}\label{sec:lsh}
We are now ready to prove \Cref{thm:LSH}. We want to solve the approximate
similarity search problem with parameters $0 < j_2 < j_1 < 1$ on a family,
$\calF$, of $n$ sets from a large universe $U$. We show a solution
that uses $O\left(n^{1 + \rho} \log(1/\eps) + \sum_{A \in \calF} \abs{A}\right)$
space and $O\left(n^{\rho} \log(1/\eps) + |Q|\right)$ query time. 
% This is already an improvement of the query time when $Q$ is large and $\eps = o(1)$.

%!TEX root = ../faster_lsh.tex

\subsection{Creating the sketches}\label{sec:LSH-creation}
We will create a data
structure similar to the LSH structure as described in \Cref{sec:lsh_intro}.
Our data structure will have parameters $L, K, M$. For each $A \in \calF$
(and query $Q$) we will create $2M \cdot L$ sketches
$\left(S'_{i, j}(A)\right)_{i \in [2M], j \in [L]}$
of size $K$ such that for any two sets $A, B \subseteq U$, $i \in [2M]$ and
$j \in [L]$ we have the following properties
\begin{itemize}
    \item $\Prp{S'_{i, j}(A) = S'_{i, j}(B)} = O(J(A, B)^K)$ \, .
    \item If $J(A, B) \geq j_1$ then
        $\Prp{S'_{i, j}(A) = S'_{i, j}(B)} = \Theta(J(A, B)^K)$ \, .
\end{itemize}
We will then combine these sketches into $M \cdot L^2$ new sketches of size $2K$
by using the tensoring technique. % as described in \Cref{sec:LSH}. %\cite{TODO:REFERENCES}, %\jbtktodo{Find the correct paper(s) to reference.}
Specifically for every $i \in [M]$ and $j, k \in [L]$ we define
\(
    S_{i, j \cdot L + k}(A) = (S'_{2i, j}(A), S'_{2i + 1, k}(A))
\)
which is clearly well-defined. By using standard hashing techniques the new
sketches can be calculated in $O(M \cdot L^2)$ time. Then for any two sets
$A, B \subseteq U$, $i \in [M]$ and $j \in [L^2]$ we have the following
properties
\begin{itemize}
    \item $\Prp{S_{i, j}(A) = S_{i, j}(B)} = O(J(A, B)^{2K})$ \, .
    \item If $J(A, B) \geq j_1$ then
        $\Prp{S_{i, j}(A) = S_{i, j}(B)} = \Theta(J(A, B)^{2K})$ \, .
\end{itemize}

By setting $K = \ceil{\frac{\log(n)}{2\log(1/j_2)}}$, $L = 6\ceil{(1/j_1)^K}$
and $M = \Theta(\log(1/\eps))$ and using the analysis of \cite{IndykM98} %\jbtktodo{Find the correct paper(s) to reference.} 
immediately give us
\[
    O\left(M L^2 + \sum_{A \in \calF} \abs{A}\right)
        = O\left(n^{1 + \rho} \log(1/\eps) + \sum_{A \in \calF} \abs{A}\right)
\]
space usage and
\begin{align*}
    O\left(M L^2 + T(\log(1 / \eps)n^{\rho/2}, \log(n), |Q|)\right)
        = O\left(n^{\rho} \log(1/\eps) + T(n^{\rho/2} \log(1 / \eps), \log(n), |Q|)\right)    
\end{align*}
query time, where $T(x, y, z)$ is the time it takes to calculate $x$ sketches
of size $y$ for a set of size $z$, and $\rho = \frac{\log(1/j_1)}{\log(1/j_2)}$.

\begin{remark}
    The parameters $K$ and $L$ differs from the usual
    analysis of the LSH framework. We have that $j_2^K \leq \frac{1}{\sqrt{n}}$ instead of
    $j_2^K \leq \frac{1}{n}$, and $n^{\rho/2} = L \geq 6 j_1^K$ instead of
    $n^{\rho} = L \geq j_1^K$. The use of tensoring entails that most of the analysis
    will be centered around the sketches
    $\left(S'_{i, j}(A)\right)_{i \in [2M], j \in [L]}$. If we
    used the usual parameters then there would be $M\ceil{\sqrt{L}}$ sketches of
    size $\ceil{\frac{K}{2}}$. This would clutter the notation making the
    analysis harder to understand.
\end{remark}

In order to create the $2M \cdot L$ sketches
$\left(S'_{i, j}(A)\right)_{i \in [2M], j \in [L]}$%, \ldots, S'_{0, L - 1}(A), \ldots, S'_{2M - 1, L - 1}(A)$
described above we first create $2M$ tables $T_{0}, \ldots, T_{2M - 1}$ of
size $L \times K$ such that for each $i \in [2M]$, $j \in [L]$ and $k \in [K]$
we have that $T_{i}[j, k]$ is an uniformly random integer chosen from the set
$\set{K S i + S j, \ldots, K S i + S(j + 1) - 1}$
where $S$ is parameter to be chosen later. The tables are independent of each
other. In every table the rows are chosen independently. Every row in every table
is filled using a source of 2-independence. Now for a given $A \subseteq U$
we do the following
\begin{enumerate}
    \item Let $S(A)$ be a size $2MKS$ similarity sketch of \Cref{sec:sketch}.
    \item For each $i \in [2M]$, $j \in [L]$ and $k \in [K]$ let
        $S'_{i, j}(A)[k] = S(A)\left[ T_i[j, k] \right]$.
\end{enumerate}
\begin{figure}
    \[
        \underbrace{
            \underbrace{
                \underbrace{|\text{\_}| \ldots |\text{\_}|}_{S}
                    \ldots
                \underbrace{|\text{\_}| \ldots |\text{\_}|}_{S}
            }_{K}
                \ldots
            \underbrace{
                \underbrace{|\text{\_}| \ldots |\text{\_}|}_{S}
                    \ldots
                \underbrace{|\text{\_}| \ldots |\text{\_}|}_{S}
            }_{K}
        }_{2M}
    \]
    \caption{The intermediate sketch $S(A)$ is first partitioned into $2M$
    segments which corresponds to the $2M$ subexperiments. Each of these
    segments then partitioned further into $K$ blocks of size $S$, which
    corresponds to the $K$ entries in the $L$ sketches in each of the
    subexperiments.}
    \label{fig:intermediateSketch}
\end{figure}
See \Cref{fig:intermediateSketch} for an intuition about how the intermediate sketch $S(A)$ is partitioned.
It takes $O(MKS \log MKS + \abs{A})$ time to create the sketch $S(A)$ by \Cref{lem:sketch_time},
and it takes $O(MLK)$ time to calculate $S'$. Thus the time needed to create sketches
$\left(S'_{i, j}(A)\right)_{i \in [2M], j \in [L]}$ for any $A \subseteq U$ is
$O\left(MLK + MKS \log MKS + \abs{A}\right)$. We will set $S = \Theta(K / j_1)$. Now the
total running time for creating the sketches 
$\left(S_{i, j}(A)\right)_{i \in [2M], j \in [L^2]}$ becomes
\[
    O\left(ML^2 + MLK + MKS \log(MKS) + \abs{A}\right)
        = O\left(n^{\rho} \log\left( 1/\eps \right) + \abs{A}\right)
\]
thus the running time of a query is $O(n^{\rho} \log\left( 1/\eps \right) + \abs{A})$.

\newtext{
The main issue is to prove that the intermediate
sketch has the desired properties despite the entries not being independent.
We call the vector of sketches $\left(S'_{i, j}(A)\right)_{j \in [L]}$ a
subexperiment for each $i \in [2M]$, and the vector of sketches 
$\left(S_{i, j}(A)\right)_{j \in [L^2]}$ an experiment for each $i \in [M]$.
It is clear that the experiment $\left(S_{i, j}(A)\right)_{j \in [L^2]}$ is
completely determined by the subexperiments $\left(S'_{2i, j}(A)\right)_{j \in [L]}$
and $\left(S'_{2i + 1, j}(A)\right)_{j \in [L]}$ for every $i \in [M]$ by construction.
We define $\calF_{bad}(Q) = \setbuilder{A \in \calF}{J(A, Q) \le j_2}$, which we will
call the family of bad sets with respect to $Q$. We also define
$\calF_{good}(Q) = \setbuilder{A \in \calF}{J(A, Q) \ge j_1}$, which we will
call the family of good sets with respect to $Q$. 
}

Furthermore it will be
useful to define the vector $\mathcal{M}'_{Q}(A)[i]$ which will be the number of matches
that the set $A \in \calF$ has in the $i \in [2M]$ subexperiment:
\[
    \mathcal{M}'_{Q}(A)[i] = \sum_{j \in [L]} \left[ S'_{i, j}(A) = S'_{i, j}(Q) \right],
\]
and the vector $\mathcal{M}_{Q}(A)[i]$ which will be the number of matches that
the set $A \in \calF$ has in the $i \in [M]$ experiment:
\[
    \mathcal{M}_{Q}(A)[i]
        = \sum_{j \in [L^2]} \left[ S_{i, j}(A) = S_{i, j}(Q) \right]
        = \mathcal{M}'_{Q}[2i]\mathcal{M}'_{Q}[2i+1]
\]

Each experiment will provide us with a lot of matches and amongst all those
matches we need to find a set $A$ with $J(A, Q) \ge j_2$, so we need to filter
away all the bad sets. To do this each experiment will choose one candidate set
amongst the matched sets. This will give us $O(\log(1/\eps))$ candidates which
we are allowed to use extra time checking by our time budget. We will use two
different techniques to choose the candidates depending on the number of matches.
If an experiment has $O(L)$ matches then we can afford to check each set using a
sketch of size $O(\max(\log n, \log(1/\eps)))$ since by using the minwise $b$-bit
hashing trick of Li et al.\cite{LiSMK11} this can be done in $O(L)$ time. If an
experiment has $\omega(L)$ matches then we pick a random match which also can
be done in $O(L)$ time.

% To prove the correctness of the algorithm  we need to show two facts that would
% be obvious and trivial if the experiments where independent but are very
% non-obvious when the experiments are correlated. We need to show that: (a)
% Only a few of the experiments has many bad matches, and (b) a good set
% $A^* \in \calF_{good}(Q)$ is matched with high probability. We need (a) to show
% that if many experiments has many matches then with high probability one of
% the randomly chosen matches are good. We need (b) to be sure that we have at
% least one good match. 

% Now to prove the correctness of the algorithm we need to prove two
% facts: (1) A good set $A^* \in \calF_{good}(Q)$ is matched with high probability,
% and (2) in only a few of the experiments do we have many bad matches, this is
% fact is needed for when we pick a random match.

We note that the experiments are conditionally independent given the intermediate
sketch, $S(Q)$, that is, if we fix the intermediate sketch then the experiments are
independent. The goal is to show that the intermediate sketch satisfies
the following with probability at least $1 - \frac{\eps}{3}$:
After fixing the intermediate sketch at least a constant fraction of the experiments
have the properties: (a) we expect
$O\left( j_2^{2K} n \right)$ bad matches in expectation, and (b) the probability
that a good set $A^{*} \in \calF_{good}(Q)$ is matched is at least $\Omega(j_1^{2K})$.
To show this we need a bit of notation, for a set $A \in \calF$ we define
\[
    Y_{i}(A) = \frac{\sum_{j \in [S]} \left[ S(A)[iS + j] = S(B)[iS + j] \right]}{S}
\]
for every $i \in [2MK]$, which corresponds to the block of the subexperiments,
and for every $l \in [2M]$ we define
$
    Z_{l}(A) = \prod_{i \in [K]} Y_{lK + i}(A)
$.
From these definition we get that
\[
    \Prpcond{S_{i, j}(A) = S_{i, j}(Q)}{Z_{2l}(A), Z_{2l + 1}(A)} = Z_{2i}(A) \cdot Z_{2i + 1}(A) \; .
\]

If the intermediate sketch is fixed then
\[
    \Epcond{\sum_{A \in \calF} \mathcal{M}_{Q}(A)[i]}{Z_{2i}, Z_{2i + 1}} = L^2 \sum_{A \in \calF_{bad}} Z_{2i}(A) \cdot Z_{2i + 1}(A)
\]
So to prove (a) we need to show that for most $i \in [M]$ then
$\sum_{A \in \calF_{bad}(Q)} Z_{2i}(A) \cdot Z_{2i + 1}(A) = O\left( j_2^{2K} n \right)$,
which formalized in the next Lemma.
\begin{lemma}\label{lem:int-bad-sets}
    Let $I \subseteq [M]$ be a set of $d$ indices, and $C$ a constant, then
    \[
        \Prp{\bigwedge_{i \in I} \left(
            \sum_{A \in \calF_{bad}(Q)} Z_{2i}(A) \cdot Z_{2i + 1}(A)
            \ge C j_2^{2K} n
        \right)} \le \left( \frac{e}{C} \right)^d \; .
    \]
\end{lemma}
\begin{proof}
    Using Markov's inequality we note that
    \begin{align*}
        &\Prp{\bigwedge_{i \in I} \left(
            \sum_{A \in \calF_{bad}} Z_{2i}(A) \cdot Z_{2i + 1}(A)
            \ge C j_2^{2K} n
        \right)}
    &\quad\le
        \Prp{\prod_{i \in I}
            \sum_{A \in \calF_{bad}} Z_{2i}(A) \cdot Z_{2i + 1}(A)
            \ge \left( C j_2^{2K} n \right)^d
        }
    \\ &\quad\le
        \frac{\Ep{\prod_{i \in I}
                \sum_{A \in \calF_{bad}} Z_{2i}(A) \cdot Z_{2i + 1}(A)}}
             {\left( C j_2^{2K} n \right)^d} \; ,
    \end{align*}
    so we just need to show that
    \[
        \Ep{\prod_{i \in I}
            \sum_{A \in \calF_{bad}} Z_{2i}(A) \cdot Z_{2i + 1}(A)}
        \le \left( e j_2^{2K} n \right)^d \; .
    \]
    To do this we enumerate $I = \set{v_0, \ldots, v_{d - 1}}$ and consider any
    $A_0, \ldots, A_{d - 1} \in \calF_{bad}(Q)$.
    \begin{align*}
            &\Ep{\prod_{i \in [d]} Z_{2v_i}(A_i) \cdot Z_{2v_i + 1}(A_i)}
        \\&\quad= 
            \prod_{i \in [d]}\Epcond{Z_{2v_i}(A_i) \cdot Z_{2v_i + 1}(A_i)}{\sigma\left( \left(Z_{2v_j}, Z_{2v_j + 1}\right)_{j \in [i - 1]} \right)}
        \\&\quad\le
            \prod_{i \in [d]}
                \left( \frac{2MKS \cdot J(A_i, Q)}{2MKS - (i - 1)2K} \right)^{2K}
        \\&\quad\le
            \left( \left( \frac{2MKS \cdot j_2}{2MKS - 2MK} \right)^{2K} \right)^d
        \\&\quad\le
            \left( \left( \frac{S \cdot j_2}{S - 1} \right)^{2K} \right)^d
        \\&\quad\le
            \left( e j_2^{2K} \right)^d
            % \\ &\qquad = 
        %     \prod_{i \in [d]} 
        %         \Epcond{[S_{v_i, l_i}(A_i) = S_{v_i, l_i}(Q)]}
        %                 {\prod_{j \in [i - 1]} [S_{v_j, l_j}(A_j) = S_{v_j, l_j}(Q)] = 1}
        % \\ &\qquad \le
        %     \prod_{i \in [d]}
        %         \left( \frac{2MKS \cdot J(A_i, Q)}{2MKS - (i - 1)2K} \right)^{2K}
        % \\ &\qquad \le
        %     \left( \left( \frac{2MKS  \cdot j'_2}{2MKS - 2MK} \right)^{2K} \right)^d
        % \\ &\qquad \le
        %     \left( \left( \frac{S \cdot j'_2}{S - 1} \right)^{2K} \right)^d
        % \\ &\qquad =
        %     \left( \left( 1 + \frac{1}{S - 1} \right)^{2K} 
        %         \cdot \left( 1 + \frac{1}{4K} \right)^{2K} 
        %         \cdot j_2^{2K}
        %     \right)^d
        % \\ &\qquad \le
        %     \left( e j_2^{2K} \right)^d
    \end{align*}
    where the first inequality follows by using \Cref{lem:moment}, and the last
    inequality follows by the definition of $S$ and $K$. Using this we see that
    \[
        \Ep{\prod_{i \in I} \sum_{A \in \calF_{bad}} Z_{2i}(A) \cdot Z_{2i + 1}(A)}
        \le \abs{\calF_{bad}(Q)}^d \cdot \left( e j_2^{2K} \right)^d
        \le \left( e j_2^{2K} n \right)^d
    \]
    which proves the result.
\end{proof}

Since $\Prpcond{S_{i, j}(A) = S_{i, j}(Q)}{Z_{2l}(A), Z_{2l + 1}(A)} = Z_{2i}(A) \cdot Z_{2i + 1}(A)$,
then to show (b): We need that for most $i \in [M]$ then $Z_{2i}(A) \cdot Z_{2i + 1}(A) \ge \Omega(j_1^{2K})$,
which is formalized in the next lemma.
\begin{lemma}\label{lem:int-good-set}
    For every set of indices $I \subseteq [M]$ of size $d$ and real number
    $\delta \in [0, 1]$, we have that
    \[
        \Prp{\bigwedge_{i \in I} \Big( Z_{2i}(A^*) \leq (1 - \delta)j_1^{K} \vee  Z_{2i + 1}(A^*) \leq (1 - \delta)j_1^{K} \Big)}
            \leq \left( 6 e^{-\delta^2 j_1 S / (32K)} \right)^d \; .
    \]
\end{lemma}
First we need the following technical lemma.
\begin{lemma}\label{lem:fixedProd}
    Let $t, k, s$ be positive integers such that $t = ks$ and let $\alpha$ be a
    real number such that $\alpha t$ is a positive integer. Let $\mathcal{B}$
    be the set of all $t$-tuples from $\set{0, 1}^t$ for which the sum of
    entries is exactly $\alpha t$, i.e.:
    \[
        \mathcal{B} = \setbuilder{(b_0, \ldots, b_t) \in \set{0, 1}^t}
                                 {\sum_{i \in [t]} b_i = \alpha t}
    \]
    Let $(a_0, \ldots, a_t)$ be drawn uniformly random from $\mathcal{B}$, and
    let $Y_i = \sum_{j \in [s]} a_{is + j}$ for every $i \in [k]$. Then for any
    real number $\delta \in [0, 1]$:
    \[
        \Prp{\prod_{i \in [k]} Y_i \le (1 - \delta)(\alpha s)^k}
            \le 2e^{-\delta^2 \alpha s / (10 k)}
    \]
\end{lemma}
\begin{proof}
    First note that we can assume that
    $2e^{-\delta^2 \alpha s / (10 k)} \le 1$ which implies that
    $5k \le \delta^2 \alpha s$.

    Let $l$ be a positive integer to be chosen later. If
    $\prod_{i \in [k]} Y_i \le (1 - \delta)(\alpha s)^k$, then one of two
    events must be true: either there exists $i \in [k]$ such that 
    $Y_i \leq \frac{\alpha s}{2}$ , or we have that
    \begin{align*}
        \prod_{i \in [k]} \left( Y_i + l \right)^{\underline{l}}
        &\le \prod_{i \in [k]} \left( Y_i + l \right)^{l}\
        \\
        &\le \left( \prod_{i \in [k]} \left( 1 + \frac{l}{Y_i} \right)^{l} \right) \left( \prod_{i \in [k]} Y_i \right)^{l}
        \\
        &\le \left( \prod_{i \in [k]} \left( 1 + \frac{l}{Y_i} \right)^{l} \right) (1 - \delta)^{l}(\alpha s)^{k l}
        \\
        &\le e^{\left( \sum_{i \in [k]} \frac{1}{Y_i} \right) l^2 - \delta l} (\alpha s)^{k l}
        \\
        &\le e^{\frac{2k}{\alpha s} l^2 - \delta l} (\alpha s)^{k l}
    \end{align*}

    For any $i \in [k]$ the probability that  $Y_i \le \frac{\alpha s}{2}$ is
    at most $e^{-\alpha s / 8}$ by standard Chernoff bound. Hence the
    probability that there exists such an $i$ is at most $k e^{- \alpha s / 8}$
    by union bound.  Now we note that
    \[
        k e^{- \alpha s / 8} 
        = k \left( e^{-\delta^2 \alpha s / (10k)} \right)^{10k/(8\delta^2)}
        \le k \left( e^{-\delta^2 \alpha s / (10k)} \right)^k
        \le e^{-\delta^2 \alpha s / (10k)}
    \]
    where the last inequality comes from the fact that $x a^x \le a$ if
    $a \le \frac{1}{2}$ and $x$ is a positive integer.

    Now we want to bound the other event. By using Markov's inequality we have
    that
    \begin{align*}
        \Prp{\prod_{i \in [k]} \left( Y_i + l \right)^{\underline{l}}
             \le e^{\frac{2k}{\alpha s} l^2 - \delta l} (\alpha s)^{kl}}
        &= \Prp{\left( \prod_{i \in [k]} \left( Y_i + l \right)^{\underline{l}} \right)^{-1}
            \ge e^{-\frac{2k}{\alpha s} l^2 + \delta l} (\alpha s)^{-kl}}
        \\ 
        &\le \Ep{\left( \prod_{i \in [k]} \left( Y_i + l \right)^{\underline{l}} \right)^{-1}}
            \cdot e^{\frac{2k}{\alpha s} l^2 - \delta l} (\alpha s)^{kl}
    \end{align*}
    Now we want to show that
    \[
        \Ep{\left( \prod_{i \in [k]} \left( Y_i + l \right)^{\underline{l}} \right)^{-1}}
        \le (\alpha s)^{-kl}
    \]

    First we define
    \[
        \mathcal{C} = 
            \setbuilder{(c_0, \ldots, c_{k - 1}) \in \set{0, \ldots, s}^k}
                       {\sum_{i \in [k]} c_i = \alpha t}
    \]
    and note that for $(c_0, \ldots, c_{k - 1}) \in \mathcal{C}$ the probability
    that $Y_i = c_i$ for all $i \in [k]$ is exactly
    \[
        \binom{t}{\alpha t}^{-1} \prod_{i \in [k]} \binom{s}{c_i} 
    \]
    Thus the expected value is
    \[
        \binom{t}{\alpha t}^{-1}
            \sum_{(c_0, \ldots, c_{k - 1}) \in \mathcal{C}}
            \prod_{i \in [k]} \left( \binom{s}{c_i} \frac{1}{(c_i + l)^{\underline{l}}} \right)
    \]
    Now we note that
    \[
        \binom{s}{c_i} \frac{1}{(c_i + l)^{\underline{l}}}
        % = \frac{s!}{c_i!(s - c_i)!} \frac{c_i!}{(c_i + l)!}
        % = \frac{s!}{(s + l)!}\frac{(s + l)!}{(s - c_i)!(c_i + l)!}
        = \frac{1}{(s + l)^{\underline{l}}} \binom{ s + l }{ c_i + l }
    \]
    Hence we get that
    \begin{align*}
        &\binom{t}{\alpha t}^{-1}
            \sum_{(c_0, \ldots, c_{k - 1}) \in \mathcal{C}}
            \prod_{i \in [k]} \left( \binom{s}{c_i} \frac{1}{(c_i + l)^{\underline{l}}} \right)
        \\
        &\qquad=
            \binom{t}{\alpha t}^{-1}
            \left( \frac{1}{(s + l)^{\underline{l}}} \right)^k
            \sum_{(c_0, \ldots, c_{k - 1}) \in \mathcal{C}}
            \prod_{i \in [k]} \binom{ s + l }{ c_i + l }
    \end{align*}
    We see that we are in fact summing over all $t$-tuples $(c_0, \ldots, c_{k - 1})$
    where $c_i \in \set{l, \ldots, l + s}$ and
    $\sum_{i \in[k]} c_i = \alpha t + kl$. So we define
    \[
        \mathcal{C}' = 
        \setbuilder{(c_0, \ldots, c_{k - 1}) \in \set{0, \ldots, s + l}^k}
                   {\sum_{i \in [k]} c_i = \alpha t + lk}
    \]
    and get that
    \begin{align*}
        &\binom{t}{\alpha t}^{-1}
            \left( \frac{1}{(s + l)^{\underline{l}}} \right)^k
            \sum_{(c_0, \ldots, c_{k - 1}) \in \mathcal{C}}
            \prod_{i \in [k]} \binom{ s + l }{ c_i + l }
        \\
        &\qquad\le
            \binom{t}{\alpha t}^{-1}
            \left( \frac{1}{(s + l)^{\underline{l}}} \right)^k
            \sum_{(c_0, \ldots, c_{k - 1}) \in \mathcal{C}'}
            \prod_{i \in [k]} \binom{ s + l}{ c_i }
        \\
        &\qquad=
            \binom{t}{\alpha t}^{-1}
            \left( \frac{1}{(s + l)^{\underline{l}}} \right)^k
            \binom{ t + kl}{\alpha t + kl}
        \\
        &\qquad=
            \frac{(t + kl)^{\underline{kl}}}{(\alpha t + kl)^{\underline{kl}}}
            \left( \frac{1}{(s + l)^{\underline{l}}} \right)^k
    \end{align*}
    It is clear that
    \[
        \left(k^l(s + l)^{\underline{l}}\right)^k
            \ge (ks + kl)^{\underline{kl}}
            = (t + kl)^{\underline{kl}}
    \]
    hence using this we get that
    \begin{align*}
        \frac{(t + kl)^{\underline{kl}}}{(\alpha t + kl)^{\underline{kl}}}
            \left( \frac{1}{(s + l)^{\underline{l}}} \right)^k
        \le \frac{k^{kl}}{(\alpha t + kl)^{\underline{kl}}}
        \le \frac{k^{kl}}{(\alpha t)^{kl}}
        = \frac{1}{(\alpha s)^{kl}} 
    \end{align*}
    So the probability is bounded by
    \[
        e^{\frac{2k}{\alpha s} l^2 - \delta l}
    \]
    which has global minimum when $l = \frac{\delta \alpha s}{4 k}$. Since $l$
    is an integer we get that the probability is bounded by
    \begin{align*}
        e^{\frac{2k}{\alpha s} \left(\frac{\delta \alpha s}{4k} \pm \frac{1}{2} \right)^2
          - \delta \left(\frac{\delta \alpha s}{4k} \pm \frac{1}{2} \right)}
        = e^{-\delta^2 \alpha s/(8 k) + k/(2 \alpha s)}
    \end{align*}
    Now using that $5k \le \delta^2 \alpha s \le \alpha s$ we get that
    $k/(2 \alpha s) \leq \delta^2 \alpha s / (50 k)$. We see that
    $-\delta^2 \alpha s/(8 k) + \delta^2 \alpha s / (50 k) \le -\delta^2 \alpha s / (10 k)$
    which finishes the proof.
\end{proof}
We can now prove that most $Z_{i}(A^*) \geq (1 - \delta)j_1^{K}$, which will
almost prove \Cref{lem:int-good-set}.
\begin{lemma}\label{lem:int-good-set-helper}
    For every set of indices $I \subseteq [2M]$ of size $d$ and real number
    $\delta \in [0, 1]$, we have that
    \[
        \Prp{\bigwedge_{i \in I} Z_{i}(A^*) \leq (1 - \delta)j_1^{K}}
            \leq \left( 3 e^{-\delta^2 j_1 S / (32K)} \right)^d
    \]
\end{lemma}
\begin{proof}
    The first we do is to write the probability as an expectation
    \[
        \Prp{\bigwedge_{l \in I} Z_l \le (1 - \delta)j_1^K}
            = \Ep{\prod_{l \in I} \left[ Z_l \le (1 - \delta)j_1^K \right]}
    \]
    For every $l \in I$ we define the random variable $\alpha_l$ by
    \[
        \alpha_l = \frac{\sum_{i \in [K]} Y_{lK + i}}{K}
    \]
    and let $P_l(\alpha)$ be the probability that $Z_l \le (1 - \delta)j_1^K$
    given that $\alpha_l = \alpha$. We then get that
    \begin{align*}
        \Ep{\prod_{l \in I} \left[ Z_l \le (1 - \delta)j_1^K \right]}
        &=
        \Ep{\Epcond{\prod_{l \in I} \left[ Z_l \le (1 - \delta)j_1^K \right]}
                   {\sigma\left( (\alpha_l)_{l \in I} \right)}}
        \\
        &=
        \Ep{\prod_{l \in I} P_l(\alpha_l)}
    \end{align*}
    Now we will split the random variable $P_l(\alpha_l)$ as follows
    \begin{align*}
        P_l(\alpha_l)
        = 
            \left[ \alpha_l \le \left(1 - \frac{\delta}{4K} \right)j_1 \right] P_l(\alpha_l)
            + \left[ \alpha_l > \left(1 - \frac{\delta}{4K} \right)j_1 \right] P_l(\alpha_l)
    \end{align*}
    If $\alpha_l > \left(1 - \frac{\delta}{4K} \right)j_1$ then
    $\alpha_l^K > \left( 1 - \frac{\delta}{4} \right)j_1^K$ so
    $\left(1 - \frac{3}{4}\delta\right)\alpha_l^K > (1 - \delta)j_1^K$. Using
    \Cref{lem:fixedProd} we get that
    \begin{align*}
        \left[ \alpha_l > \left(1 - \frac{\delta}{4K} \right)j_1 \right] P_l(\alpha_l)
        &\le
        \left[ \alpha_l > \left(1 - \frac{\delta}{4K} \right)j_1 \right]
            2e^{-\delta^2 \alpha_l S 3^2 / (4^2 \cdot 10 K)}
        \\
        &\le
        2e^{-\delta^2 j_1 S 3^3 / (4^3 \cdot 10 K)}
        \\
        &\le
        2e^{-\delta^2 j_1 S / (32 K)}
    \end{align*}
    where the second inequality uses that
    $\alpha_l > \left(1 - \frac{\delta}{4K} \right)j_1 \ge \frac{3}{4}j_1$ and last
    inequality uses that $\frac{3^3}{4^3 \cdot 10} > \frac{1}{32}$.

    Now using this we get that
    \begin{align*}
        \Ep{\prod_{l \in I} P_l(\alpha_l)}
        \le \Ep{\prod_{l \in I} \left( 
            \left[ \alpha_l \le \left(1 - \frac{\delta}{4K} \right)j_1 \right]
            + 2e^{-\delta^2 j_1 s / (32 k)}
        \right)}
    \end{align*}
    Now for every subset $I' \subseteq I$ of size $d'$ we have that
    \begin{align*}
        \Ep{\prod_{l \in I'} \left[ \alpha_l \le \left(1 - \frac{\delta}{4K} \right)j_1 \right]}
        &\le
        \Prp{\frac{\sum_{l \in I'} \alpha_l}{d'} \le \left(1 - \frac{\delta}{4K} \right)j_1}
        \\
        &=
        \Prp{\frac{\sum_{l \in I'} \sum_{i \in [K]} Y_{lK + i}}{d' t} \le \left(1 - \frac{\delta}{4K} \right)j_1}
        \\
        &\le
        e^{-\delta^2 j_1 d' S / (32K)}
        \\
        &=
        \left( e^{-\delta^2 j_1 S / (32K)} \right)^{d'}
    \end{align*}
    where the inequality uses that \Cref{lem:chernoffLemma}. Hence we get that
    \[
        \Prp{\bigwedge_{l \in I} Z_l \le (1 - \delta)j_1^K}
            \leq \left( 3 e^{-\delta^2 j_1 S / (32K)} \right)^d
    \]
    as we wanted.
\end{proof}
We are now ready to prove \Cref{lem:int-good-set} which follows easily from
\Cref{lem:int-good-set-helper}.
\begin{proof}[Proof of \Cref{lem:int-good-set}]
    By union bound and \Cref{lem:int-good-set-helper} we get that
    \begin{align*}
        &\Prp{ \bigwedge_{i \in I} \left( Z_{2i}(A^*) \leq (1 - \delta)j_1^{K} \vee  Z_{2i + 1}(A^*) \leq (1 - \delta)j_1^{K} \right)}
            \\
            &\quad\le \sum_{(v_i)_{i \in I} \in \set{0, 1}^d} \Prp{ \bigwedge_{i \in I} \left(
                Z_{2i + v_i}(A^*) \leq (1 - \delta)j_1^{K}
            \right)}
            \\
            &\quad\le \sum_{(v_i)_{i \in I} \in \set{0, 1}^d}
                \left( 3 e^{-\delta^2 j_1 S / (32K)} \right)^d
            \\
            &\quad= \left( 6 e^{-\delta^2 j_1 S / (32K)} \right)^d
    \end{align*}
\end{proof}

Having \Cref{lem:int-bad-sets} and \Cref{lem:int-good-set} it becomes easy to
prove the main theorem of this section.
\begin{theorem}\label{thm:good-experiments}
    Let $A \in \calF_{good}(Q)$ and,
    $M_{good} \subseteq [M]$ be the set of experiments, such that, for every
    $i \in M_{good}$
    \begin{enumerate}
        \item $
            \sum_{A \in \calF_{bad}(Q)} Z_{2i}(A) \cdot Z_{2i + 1}(A)
                \le 4 e j_2^{2K} n
        $ \; ,
        \item $
            Z_{2i}(A^*) \ge \frac{1}{2}j_1^{K} \wedge  Z_{2i + 1}(A^*) \ge \frac{1}{2}j_1^{K}
        $ \; .
    \end{enumerate}
    Then $\abs{M_{good}} \ge \frac{M}{3}$ with probability at least $1 - \frac{\eps}{2}$.
\end{theorem}
\begin{proof}
    % By \Cref{lem:int-bad-sets} and \Cref{cor:chernoff} we have Chernoff-type
    % bounds on the number of experiments which does not satisfy the first requirement.
    % Choosing $M$ large enough we get that at most $\frac{M}{3}$ experiments does
    % not satisfy the first requirement with probability at least $1 - \frac{\eps}{4}$.
    By \Cref{lem:int-bad-sets}, \Cref{lem:int-good-set}, and a standard analysis of Chernoff bounds, we have Chernoff-type bounds on the number of experiments which does not satisfy either of the requirements.
    Now choosing $S$ large enough we get that $6 e^{-\frac{1}{2}^2 j_1 S / (32K)} \le \frac{1}{4}$,
    so choosing $M$ large enough at most $\frac{M}{3}$ experiments does
    not satisfy the first requirement with probability at least $1 - \frac{\eps}{4}$,
    and at most $\frac{M}{3}$ experiments
    does not satisfy the second requirement  with probability at least $1 - \frac{\eps}{4}$.
    Now a union bound finishes the proof.
\end{proof}

\subsection{Filtering}
Using \Cref{thm:good-experiments},
we get that good set $A^{*} \in \calF_{good}(Q)$ is matched with probability
at least $1 - \frac{3}{4}\eps$. Now it is not enough to know that a good set $A^{*} \in \calF_{good}(Q)$ is matched
in one of the experiments with high probability, we also need to be able to
find the good set amongst all the matched sets. To do this we will apply a
two-tiered filtering algorithm. First we use a fast and imprecise filtering step
to select a candidate $C_i \in \calF$ amongst the matched set in the
$i$'th experiment for every experiment $i \in [M]$. This will give us $M$
candidates $C_0, \ldots, C_{M - 1}$ for which will use a slower and more
precise filtering step to select the actual set. We will show that this
two-tiered filtering algorithm finds a good set with high probability if such
a set exists.

In the first filtering step we will use two different approaches depending on
the number of matches in an experiment. If there is more than $2CL^2$ matches in
a experiment then we will choose a random match. If there is less than $2CL^2$
matches then we will check every match using a sketch of size
$\Theta(\max(\log n, \log(1/\eps)))$, this can be done in constant time per
element by using minwise $b$-bit hashing, and pick the first element that is above a
certain threshold. If no element is above the threshold then we will for the
sake of the analysis assume that the candidate picked is the empty set.

In the second filtering step we have a set of $M = O(\log(1/\eps))$ candidates,
which we will check using a sketch of size $\Theta(\log^2(n) \log(1/\eps))$, again
using minwise $b$-bit hashing this can be done in $\Theta(\log^2(n))$ time per element. This
allows us distinguish between elements with similarity less than $j_2$ and
elements with similarity at least $j'_2$. Again we pick the first element that
is above a certain threshold. %The goal is of course to pick the set with
% similarity at least $j'_2$ which advanced from the first filtering step.

We define
\[
    \calM_{Q}(A)[i] = \sum_{j \in [L]} \left[ S_{i, j}(A) = S_{i, j}(Q) \right]
\]
to be the number of matches of $A \in \calF$ in the $i$'th experiment where $i \in [M]$.
Now let 
\[
    I_{Q}(A^{*}) = \setbuilder{i \in M_{good}}
    {
        \left( \calM_{Q}(A^{*})[i] > 0 \right) 
        \wedge
        \left( \sum_{A \in \calF_{bad}(Q)} \calM_{Q}(A)[i] \le CL^2 \right)
    }
\]
be the set of experiments where $A^{*} \in \calF_{good}(Q)$ is matched and
where there are not to many bad sets are matched. Using
\Cref{thm:good-experiments} and the analysis from Dahlgaard et al.\cite{DahlgaardKT17}
we get that $\abs{I_{Q}(A^{*})} = \Omega\left(\log\left(1/\eps\right)\right)$
with probability at least $1 - \frac{3}{4}\eps$ by choosing 
$M = \Theta(\log(1/\eps))$ and $C = \Theta(1)$ large enough.

We want to show that at least one of the experiments chooses a candidate with
similarity at least $j'_2$, in particular we will show that
\[
    \Prp{\bigwedge_{i \in I_{Q}(A^{*})} \left( J(C_i, Q) < j'_2 \right)}
        \le \left( \frac{1}{2} \right)^{\abs{I_{Q}(A^{*})}}
\]
Since we will filter in two different ways depending on the number of matches,
we will split $I_{Q}(A^*)$ into two:
The experiments $I'_{Q}(A^*)$ with many matches
\[
    I'_{Q}(A^*) = \setbuilder{i \in I_{Q}(A^{*})}{\sum_{A \in \calF} \calM_{Q}(A)[i] > 2CL^2}
\]
and the experiments $I''_{Q}(A^*)$ with few matches
\[
    I''_{Q}(A^{*}) = \setbuilder{i \in I_{Q}(A^{*})}{\sum_{A \in \calF} \calM_{Q}(A)[i] \le 2CL^2}
\]

First we consider the case where the is at least one good experiment with few
matches. As mentioned we will create a sketch $T(Q)$ of size
$t = \Theta(\max(\log n, \log(1/\eps)))$ using the minwise $1$-bit hashing trick.
We then get that
\[
    \Prp{T(Q)[i] = T(A)[i]}
        = J(A, Q) + (1 - J(A, Q))\frac{1}{2}
        = \frac{1 + J(A, Q)}{2}
\]
for any $A \in \calF$ and $i \in [t]$. We pick the threshold
$\gamma = \frac{1 + \frac{j_1 + j'_2}{2}}{2}$. Using Hoeffding's inequality we
get that for $A \in \calF_{good}(Q)$ then
\[
    \Prp{\sum_{i \in [t]} [T(A)[i] = T(Q)[i]] \le \gamma t}
        \le \frac{\eps}{8} \cdot \frac{1}{n}
\]
and for $A \in \calF_{bad}(Q)$ then
\[
    \Prp{\sum_{i \in [t]} [T(A)[i] = T(Q)[i]] \ge \gamma t}
        \le \frac{\eps}{8} \cdot \frac{1}{n}
\]
So a union bound over all the sets in all the experiments, for which there are
at most $O(CL^2 \log(1/\eps))$, shows that the probability that none of the candidates
has similarity at least $j'_2$ is at most $\frac{\eps}{8}$. It takes $O(1)$
time to check each set so it takes $O(L^2)$ time to filter each experiment.

Now we assume that every good experiment has many matches. We note that
since we pick a random element in every experiment independently then we get that
\[
    \Prp{\bigwedge_{i \in I'_{Q}(A^{*})} \left( J(C_i, Q) < j'_2 \right)}
        = \prod_{i \in I'_{Q}(A^{*})} \Prp{J(C_i, Q) < j'_2}
\]
Using Markov's inequality we get that
\[
    \Prp{J(C_i, Q) < j'_2} \le \frac{1}{2}
\]
for $i \in I'{Q}(A^{*})$ since we know that
$\sum_{A \in \calF_{bad}(Q)} \calM_{Q}(A)[i] \le CL^2$. This shows that
\[
    \Prp{\bigwedge_{i \in I'_{Q}(A^{*})} \left( J(C_i, Q) < j'_2 \right)}
        \le \left( \frac{1}{2} \right)^{\abs{I'_{Q}(A^*)}}
\]
Since every experiment has many matches this implies that the probability that
none of the candidates has similarity at least $j'_2$ is at most $\frac{\eps}{8}$.
This shows that the probability that the first step of filtering fails is at
most $\frac{\eps}{8}$. 

We have now shown that the probability, that none of the candidates has Jaccard
similarity at least $j'_2$, is at most $\frac{7}{8}\eps$. Now to find the
correct candidate we will create a sketch $T'(Q)$ of size
$t' = \Theta(\log^2(n)\log(1/\eps))$ using the minwise $1$-bit hashing trick.
We pick the threshold $\gamma' = \frac{1 + \frac{j_2 + j'_2}{2}}{2}$.
Using Hoeffding's inequality we get that for candidates $C$ with
$J(C, Q) \ge j'_2$  then
\[
    \Prp{\sum_{i \in [t]} [T(C)[i] = T(Q)[i]] \le \gamma' t'}
        \le \frac{\eps^2}{8}
\]
and for candidates $C$ with $J(C, Q) \le j_2$ then
\[
    \Prp{\sum_{i \in [t]} [T(C)[i] = T(Q)[i]] \ge \gamma' t'}
        \le \frac{\eps^2}{8}
\]
A union bound over all $O(\log(1/\eps))$ candidates shows that the probability,
that the set chosen has similarity less than $j_2$, is at most $\frac{\eps}{8}$.
Checking each candidate takes $O(\log^2(n))$ time so we use a total of
$O(\log^2(n)\log(1/\eps))$ time for second filtering step.

Combining all the steps shows that the data structure finds a set $A \in \calF$
with $J(A, Q) \ge j_2$ if there exists a set $B \in \calF$ with
$J(B, Q) \ge j_1$ with probability at least $1 - \eps$. Now since we want the
data structure to only have a 1-sided error then we calculate the exact Jaccard
similarity of the set in $O(\abs{Q})$ time.
}

\newpage
\bibliographystyle{plain}
\bibliography{simsketch.bib}

\begin{thebibliography}{10}

\bibitem{AndoniI06}
Alexandr Andoni and Piotr Indyk.
\newblock Efficient algorithms for substring near neighbor problem.
\newblock In {\em Proceedings of the Seventeenth Annual {ACM-SIAM} Symposium on
  Discrete Algorithms, {SODA} 2006, Miami, Florida, USA, January 22-26, 2006},
  pages 1203--1212. {ACM} Press, 2006.

\bibitem{BachrachP13}
Yoram Bachrach and Ely Porat.
\newblock Sketching for big data recommender systems using fast pseudo-random
  fingerprints.
\newblock In Fedor~V. Fomin, Rusins Freivalds, Marta~Z. Kwiatkowska, and David
  Peleg, editors, {\em Automata, Languages, and Programming - 40th
  International Colloquium, {ICALP} 2013, Riga, Latvia, July 8-12, 2013,
  Proceedings, Part {II}}, volume 7966 of {\em Lecture Notes in Computer
  Science}, pages 459--471. Springer, 2013.

\bibitem{BayardoMS07}
Roberto~J. Bayardo, Yiming Ma, and Ramakrishnan Srikant.
\newblock Scaling up all pairs similarity search.
\newblock In {\em Proceedings of the 16th International Conference on World
  Wide Web}, WWW '07, page 131–140, New York, NY, USA, 2007. Association for
  Computing Machinery.

\bibitem{Broder97}
A.~Broder.
\newblock On the resemblance and containment of documents.
\newblock In {\em Proceedings of the Compression and Complexity of Sequences
  1997}, SEQUENCES '97, pages 21--29, Washington, DC, USA, 1997. IEEE Computer
  Society.

\bibitem{BCFM00}
Andrei~Z. Broder, Moses Charikar, Alan~M. Frieze, and Michael Mitzenmacher.
\newblock Min-wise independent permutations.
\newblock {\em J. Comput. Syst. Sci.}, 60(3):630--659, 2000.
\newblock Announced at STOC'98.

\bibitem{BroderGMZ97}
Andrei~Z. Broder, Steven~C. Glassman, Mark~S. Manasse, and Geoffrey Zweig.
\newblock Syntactic clustering of the web.
\newblock {\em Computer Networks}, 29(8-13):1157--1166, 1997.

\bibitem{Christiani17}
Tobias Christiani.
\newblock Fast locality-sensitive hashing for approximate near neighbor search.
\newblock {\em CoRR}, abs/1708.07586, 2017.

\bibitem{CK07}
Edith Cohen and Haim Kaplan.
\newblock Summarizing data using bottom-k sketches.
\newblock In {\em Proc. 26th PODC}, pages 225--234, 2007.

\bibitem{DahlgaardKRT15}
S{\o}ren Dahlgaard, Mathias B{\ae}k~Tejs Knudsen, Eva Rotenberg, and Mikkel
  Thorup.
\newblock Hashing for statistics over k-partitions.
\newblock In Venkatesan Guruswami, editor, {\em {IEEE} 56th Annual Symposium on
  Foundations of Computer Science, {FOCS} 2015, Berkeley, CA, USA, 17-20
  October, 2015}, pages 1292--1310. {IEEE} Computer Society, 2015.

\bibitem{DahlgaardKT17}
S{\o}ren Dahlgaard, Mathias B{\ae}k~Tejs Knudsen, and Mikkel Thorup.
\newblock Fast similarity sketching.
\newblock In Chris Umans, editor, {\em 58th {IEEE} Annual Symposium on
  Foundations of Computer Science, {FOCS} 2017, Berkeley, CA, USA, October
  15-17, 2017}, pages 663--671. {IEEE} Computer Society, 2017.

\bibitem{HIM12}
Sariel Har{-}Peled, Piotr Indyk, and Rajeev Motwani.
\newblock Approximate nearest neighbor: Towards removing the curse of
  dimensionality.
\newblock {\em Theory of Computing}, 8(1):321--350, 2012.

\bibitem{Henzinger06}
Monika Henzinger.
\newblock Finding near-duplicate web pages: A large-scale evaluation of
  algorithms.
\newblock In {\em Proceedings of the 29th Annual International ACM SIGIR
  Conference on Research and Development in Information Retrieval}, SIGIR '06,
  page 284–291, New York, NY, USA, 2006. Association for Computing Machinery.

\bibitem{HT22:chaos}
Jakob B{\ae}k~Tejs Houen and Mikkel Thorup.
\newblock Understanding the moments of tabulation hashing via chaoses.
\newblock In {\em 49th International Colloquium on Automata, Languages, and
  Programming, {ICALP} 2022, July 4-8, 2022, Paris, France}, volume 229 of {\em
  LIPIcs}, pages 74:1--74:19. Schloss Dagstuhl - Leibniz-Zentrum f{\"{u}}r
  Informatik, 2022.

\bibitem{IndykM98}
Piotr Indyk and Rajeev Motwani.
\newblock Approximate nearest neighbors: Towards removing the curse of
  dimensionality.
\newblock In Jeffrey~Scott Vitter, editor, {\em Proceedings of the Thirtieth
  Annual {ACM} Symposium on the Theory of Computing, Dallas, Texas, USA, May
  23-26, 1998}, pages 604--613. {ACM}, 1998.

\bibitem{Li15}
Ping Li.
\newblock 0-bit consistent weighted sampling.
\newblock In {\em Proceedings of the 21th ACM SIGKDD International Conference
  on Knowledge Discovery and Data Mining}, pages 665--674, 08 2015.

\bibitem{LiOZ12}
Ping Li, Art~B. Owen, and Cun{-}Hui Zhang.
\newblock One permutation hashing.
\newblock In Peter~L. Bartlett, Fernando C.~N. Pereira, Christopher J.~C.
  Burges, L{\'{e}}on Bottou, and Kilian~Q. Weinberger, editors, {\em Advances
  in Neural Information Processing Systems 25: 26th Annual Conference on Neural
  Information Processing Systems 2012. Proceedings of a meeting held December
  3-6, 2012, Lake Tahoe, Nevada, United States.}, pages 3122--3130, 2012.

\bibitem{LiSMK11}
Ping Li, Anshumali Shrivastava, Joshua~L. Moore, and Arnd~Christian
  K{\"{o}}nig.
\newblock Hashing algorithms for large-scale learning.
\newblock In John Shawe{-}Taylor, Richard~S. Zemel, Peter~L. Bartlett, Fernando
  C.~N. Pereira, and Kilian~Q. Weinberger, editors, {\em Advances in Neural
  Information Processing Systems 24: 25th Annual Conference on Neural
  Information Processing Systems 2011. Proceedings of a meeting held 12-14
  December 2011, Granada, Spain.}, pages 2672--2680, 2011.

\bibitem{MRS08}
Christopher~D. Manning, Prabhakar Raghavan, and Hinrich Sch{\"{u}}tze.
\newblock {\em Introduction to information retrieval}.
\newblock Cambridge University Press, 2008.

\bibitem{ODonnellWZ14}
Ryan O'Donnell, Yi~Wu, and Yuan Zhou.
\newblock Optimal lower bounds for locality-sensitive hashing (except when q is
  tiny).
\newblock {\em {TOCT}}, 6(1):5:1--5:13, 2014.

\bibitem{patrascu10kwise-lb}
Mihai P{\v{a}}tra{\c{s}}cu and Mikkel Thorup.
\newblock On the k-independence required by linear probing and minwise
  independence.
\newblock In Samson Abramsky, Cyril Gavoille, Claude Kirchner, Friedhelm Meyer
  auf~der Heide, and Paul~G. Spirakis, editors, {\em Automata, Languages and
  Programming}, pages 715--726, Berlin, Heidelberg, 2010. Springer Berlin
  Heidelberg.

\bibitem{SchmidtSS95}
Jeanette~P. Schmidt, Alan Siegel, and Aravind Srinivasan.
\newblock Chernoff-hoeffding bounds for applications with limited independence.
\newblock {\em {SIAM} J. Discrete Math.}, 8(2):223--250, 1995.

\bibitem{ShakhnarovichDI08}
Gregory Shakhnarovich, Trevor Darrell, and Piotr Indyk.
\newblock {\em Nearest-Neighbor Methods in Learning and Vision: Theory and
  Practice (Neural Information Processing)}.
\newblock The MIT Press, 2006.

\bibitem{Shrivastava17}
Anshumali Shrivastava.
\newblock Optimal densification for fast and accurate minwise hashing.
\newblock In Doina Precup and Yee~Whye Teh, editors, {\em Proceedings of the
  34th International Conference on Machine Learning, {ICML} 2017, Sydney, NSW,
  Australia, 6-11 August 2017}, volume~70 of {\em Proceedings of Machine
  Learning Research}, pages 3154--3163. {PMLR}, 2017.

\bibitem{ShrivastavaLiICML14}
Anshumali Shrivastava and Ping Li.
\newblock Densifying one permutation hashing via rotation for fast near
  neighbor search.
\newblock In {\em Proceedings of the 31th International Conference on Machine
  Learning, {ICML} 2014, Beijing, China, 21-26 June 2014}, volume~32 of {\em
  {JMLR} Workshop and Conference Proceedings}, pages 557--565. JMLR.org, 2014.

\bibitem{ShrivastavaLiUAI14}
Anshumali Shrivastava and Ping Li.
\newblock Improved densification of one permutation hashing.
\newblock In Nevin~L. Zhang and Jin Tian, editors, {\em Proceedings of the
  Thirtieth Conference on Uncertainty in Artificial Intelligence, {UAI} 2014,
  Quebec City, Quebec, Canada, July 23-27, 2014}, pages 732--741. {AUAI} Press,
  2014.

\bibitem{thorup13bottomk}
Mikkel Thorup.
\newblock Bottom-k and priority sampling, set similarity and subset sums with
  minimal independence.
\newblock In {\em Proceedings of the Forty-Fifth Annual ACM Symposium on Theory
  of Computing}, STOC '13, page 371–380, New York, NY, USA, 2013. Association
  for Computing Machinery.

\end{thebibliography}

% \newpage
% \appendix
% \input{appendix}

\end{document}